%% file: main.tex
\documentclass{paper}

\usepackage{amsmath,amssymb,amsthm,graphicx}
\usepackage{yhmath}
\usepackage{enumerate}

\usepackage[utf8]{inputenc}
\usepackage[T1]{fontenc}
\usepackage{dcolumn}
\usepackage{bm}
\usepackage{appendix}

\usepackage{lipsum}
\usepackage{mathtools}
\usepackage{cuted}

\usepackage{subcaption}
\graphicspath{{img/},{fig/}}

\newcommand{\sech}{\operatorname{sech}}
\newcommand{\bu} {{\bf u}}
\newcommand{\bv} {{\bf v}}
\newcommand{\bs} {{\bf s}}

\newcommand\egaldef{\stackrel{\mbox{\upshape\tiny def}}{=}}

\newcommand\1{\leavevmode\hbox{\rm \small1\kern-0.35em\normalsize1}}
\newcommand\ind[1]{\1_{\{#1\}}}
\newcommand\EE{\mathsf{E}}
\def\DD{\displaystyle}

\DeclareMathOperator*{\eg}{=}

\newtheorem{prop}{Proposition}[section]

\newtheorem{lem}[prop]{Lemma}

\begin{document}

\title{Thermodynamics of Restricted Boltzmann Machines and Related Learning Dynamics}
\author{A. Decelle \and G. Fissore \and C. Furtlehner}



\maketitle

\abstract{
We investigate the thermodynamic properties of a Restricted Boltzmann Machine (RBM), a simple energy-based generative model used in the context of unsupervised learning.
Assuming the information content of this model to be mainly reflected by the spectral properties of its weight matrix $W$, we try to make a realistic analysis by averaging over an appropriate statistical ensemble of RBMs. 

First, a phase diagram is derived. Otherwise similar to that of the Sherrington-Kirkpatrick (SK) model with ferromagnetic couplings, the RBM's phase diagram presents a ferromagnetic phase which may or may not be of compositional type depending on the kurtosis of the distribution of the components of the singular vectors of $W$.

Subsequently, the learning dynamics of the RBM is studied in the thermodynamic limit. 
A ``typical'' learning trajectory is shown to solve an effective dynamical equation, based on the aforementioned ensemble average and explicitly involving order parameters obtained from the thermodynamic analysis. In particular, this let us show how the evolution of the dominant singular values of $W$, and thus of the unstable modes, is driven by the input data. At the beginning of the training, in which the RBM is found to operate in the linear regime, the unstable modes reflect the dominant covariance modes of the data. In the non-linear regime, instead, the selected modes interact and eventually impose a matching of the order parameters to their empirical counterparts estimated from the data.

Finally, we illustrate our considerations by performing experiments on both artificial and real data, showing in particular how the RBM operates in the ferromagnetic compositional phase.
 
}


\section{Introduction}
The Restricted Boltzmann Machine (RBM)~\cite{Smolensky} is an important machine learning tool used in
many applications, by virtue of its ability to model complex probability distributions.
It is a neural network which serves as a generative model, in the sense that it is able to approximate the probability distribution corresponding to the empirical distribution of any set of high-dimensional data points living in a discrete or real space of dimension $N \gg 1$.
From the theoretical point of view, the RBM is of high interest as it is one of the simplest neural network generative models and the probability distribution that it defines presents a simple analytic form.
Moreover, there are clear connections between RBMs and well known disordered systems in statistical physics. As an example, when data are composed by vectors with binary components the discrete RBM takes the form of an heterogeneous Ising model composed of one layer of visible units (the observable variables) connected to one layer of hidden units
(the latent or hidden variables building up the dependencies between the visible ones),
in which couplings and fields are obtained from the training data through a learning procedure.
In order to build more powerful models, RBMs can be stacked to form ``deep'' architectures. In such a case, they can form a multi-layer generative model known as a Deep Boltzmann Machine (DBM) ~\cite{salakhutdinov2009deep} or they can be stacked and trained layerwise as a pre-training procedure for neural networks ~\cite{HiSa}.
The standard learning algorithms in use are the contrastive divergence~\cite{Hinton_CD} (CD) and the refined Persistence CD~\cite{Tieleman} (PCD), which are based on a quick Monte Carlo estimation of the response function of the RBM and are efficient and well documented~\cite{Hinton_guide}. Nevertheless, despite some interesting
interpretations of CD in terms of non-equilibrium statistical physics~\cite{Salazar},
the learning of RBMs remains a set of obscure recipes from the statistical physics point of view: hyperparameters (like the size of the
hidden layer) are supposed to be set empirically without any theoretical guidelines.

Historically, statistical physics played a central role in studying the theoretical foundations of neural networks. In particular, during the 1980s many works on the Hopfield model~\cite{Hopfield,AmGuSo3,Gardner,Derrida-Gardner} managed to define its learning capacity and to compute the number of independent patterns that it could store.
It is worth noticing that, as RBMs are ultimately defined as a Boltzmann distribution with pairwise interactions on a bipartite graph, they can be studied in a way similar to that used for the Hopfield model. The analogy is even stronger since connections between the Hopfield model and RBMs have been made explicit when using Gaussian hidden variables~\cite{BarraBSC12}, here the number of patterns of the Hopfield model corresponding to the number of hidden units.
Motivated by a renewed excitement for neural networks, recent works actually propose to exploit the statistical physics formulation of the RBM to understand what is its learning capacity and how mean-field methods can be exploited to improve the model. In~\cite{TAP_train,huang2015advanced,takahashi2016mean}, mean-field based learning methods using TAP equations are developed. TAP solutions are usually expected to define a decomposition of the 
measure in terms of pure thermodynamical states and are useful both as an algorithm to compute the marginals of the variables of the model and to identify the pure states when they are yet unknown. For instance, in a sparse explicit Boltzmann machine (i.e. without latent variables) this implicit clustering 
can be done by means of belief propagation fixed points~\footnote{a somewhat different form of the TAP equations} with simple empirical learning rules~\cite{FuLaAu}. 
In~\cite{Barra,HHuang}, an analysis of the static properties of RBMs is done assuming a given weight matrix $W$, in order to understand collective phenomena 
in the latent representation, i.e. the way latent variables organize themselves in a compositional phase~\cite{Agliari,TuMo} to represent actual data. 
These analysis make use of the replica trick (or equivalent) making the common
assumption that the components of the weight matrix $W$ are i.i.d.; despite the fact that this approach may give some insights into the retrieval phase, 
this approximation is problematic since, as far as a realistic RBM is concerned (an RBM learned on data), 
the learning mechanism introduces correlations within the weights of $W$ and then it seems rather crude to continue to assume the independence and hope 
to understand the realistic statistical properties of the model.

Concerning the learning procedure of neural networks, many recent statistical physics based analyses have been proposed,
most of them within teacher-student setting~\cite{ZdKr}. This imposes a rather strong assumption on the data in the sense that
it is assumed that these are generated from
a model belonging to the parametric family of interest, hiding as a consequence the role played by the data themselves in the procedure.
From the analysis of related linear models~\cite{TiBi,Bou-Kamp}, it is already a well established fact
that a selection of the most important modes of the singular values
decomposition (SVD) of the data is performed in the linear case.
In fact in the simpler context of linear feed-forward models the learning dynamics can be fully characterized by means of the
SVD of the data matrix~\cite{Ganguli}, showing in particular the emergence of each mode by order of importance with respect to the corresponding singular values.

First steps to follow this guideline have been done in~\cite{DeFiFu}, in the context of a general RBM and to address the shortcomings of previous analyses,
in particular concerning the assumptions over the weights distribution. To this end it has been proposed 
to characterize both the learned RBM and the learning process itself by means of the SVD spectrum of the weight matrix in order to 
single out the information content of the RBM.
It is assumed that the SVD spectrum is split in a continuous bulk of singular vectors corresponding to noise and a set of outliers that represent the information content.
By doing this it is possible to go beyond the usual unrealistic assumption of i.i.d. weights made for analyzing RBMs.
Proceeding along this direction, in the present work we first present a thermodynamic analysis of RBMs under the more realistic assumptions over the weight matrix that we propose. 
Then, on the same basis, the learning dynamics of RBMs is studied by direct analysis of the dynamics of the SVD modes, both in the linear and non-linear regimes.

\begin{figure}[ht]
\centerline{\resizebox*{0.7\textwidth}{!}{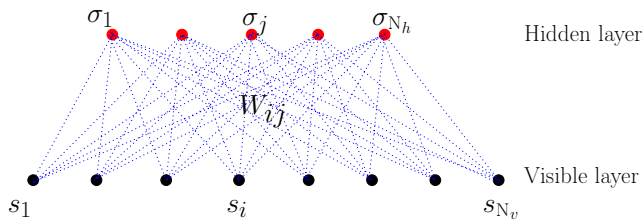}}
\caption{\label{fig:rbm} bipartite structure of the RBM.}
\end{figure}

The paper is organized as follows:
in Section~\ref{sec:rbm} we introduce the RBM model and its associated learning procedures.
Section~\ref{sec:thermo} presents the static thermodynamical properties of the RBM
with realistic hypothesis on its weights:
a statistical ensemble of weight matrices is discussed in Section~\ref{sec:stat-ensemble}; mean-field equations in the replica-symmetric (RS)
framework are given in Section~\ref{sec:RS} and the corresponding phase diagram is
studied in Section~\ref{sec:phasediag} with a proper delimitation of the RS domain where the learning procedure is supposed to take place.
The ferromagnetic phase is studied in great details in~\ref{sec:ferro_phase} by looking in particular at the conditions leading
to a compositional phase. Section~\ref{sec:rbmdyn} is devoted to the learning dynamics.
In Section~\ref{sec:equations}, a deterministic learning equation is derived in the thermodynamic limit and a set of dynamical parameters is shown to emerge naturally from the SVD of the weight matrix.
This equation is analyzed for linear RBMs in Section~\ref{sec:linear} in order to identify the unstable
deformation modes of $W$ that result in the first emerging patterns 
at the beginning of the learning process; the non-linear regime is described in Section~\ref{sec:non-linear},
on the basis of the thermodynamic analysis, by numerically solving the effective learning equations in simple cases.
Our analysis is finally illustrated and validated in Section~\ref{sec:validation} by actual tests on the MNIST dataset.

\section{The RBM and its associated learning procedure}\label{sec:rbm}
An RBM is a Markov random field with pairwise interactions defined on a bipartite graph formed by two layers of non-interacting variables: 
the visible nodes and the hidden nodes representing respectively data configurations and latent representations (see Figure~\ref{fig:rbm}). 
The former noted $\bm{s} = \{s_i,i=1\ldots N_v\}$ 
correspond to explicit representations of the data while the latter noted $\bm{\sigma} = \{\sigma_j,j=1\ldots N_h\}$ 
are there to build arbitrary dependencies among the visible units. They play the role of an interacting field among visible nodes. 
Usually the nodes are binary-valued (of Boolean type or Bernoulli distributed) but Gaussian distributions or more broadly arbitrary distributions on real-valued bounded support are also used~\cite{general_RBM}, ultimately making RBMs adapted to more heterogeneous data sets. 
Here to simplify we assume that visible and hidden nodes will be taken as binary variables $s_i,\sigma_j \in \{-1,1\}$ 
(using $\pm 1$ values gives the advantage of working with symmetric equations hence avoiding to deal with the ``hidden'' biases on the variables that appear when considering binary $\{0,1\}$ variables).
Like in the Hopfield model~\cite{Hopfield}, which can actually be cast into an RBM~\cite{BarraBSC12}, an energy function is defined for a configuration of nodes
\begin{equation}
    E(\bm{s},\bm{\sigma}) = -\sum_{i,j} s_i W_{ij} \sigma_j + \sum_{i=1}^{N_v} \eta_i s_i + \sum_{j=1}^{N_h} \theta_j \sigma_j
    \label{eq_ener_rbm}
\end{equation}
and this is exploited to define a joint distribution
between visible and hidden units, namely the Boltzmann distribution

\begin{equation}
    p(\bm{s},\bm{\sigma}) = \frac{e^{-E(\bm{s},\bm{\sigma})}}{Z}
    \label{eq_proba_rbm}
\end{equation}

\noindent where  $W$ is the weight matrix and $\bm{\eta}$ and $\bm{\theta}$ are biases, or external fields on the variables.
\( \textstyle Z = \sum_{\bm{s},\bm{\sigma}} e^{-E(\bm{s},\bm{\sigma})} \) is the partition function of the system. The joint distribution between visible
variables is then obtained by summing over hidden ones.
In this context, learning the parameters of the RBM means that, given a dataset of $M$ samples composed of 
$N_v$ variables, we ought to infer values to $W$, $\bm{\eta}$ and $\bm{\theta}$ such that new generated data obtained by sampling this distribution 
should be similar to the input data. The general method to infer the parameters is to maximize the log likelihood of the model, where the pdf (\ref{eq_proba_rbm}) has 
first been summed over the hidden variables
\begin{equation}\label{eq:LL}
  \mathcal{L} = \sum_j \langle \log(2 \cosh(\sum_i W_{ij} s_i - \theta_j))\rangle_{\rm Data} -\sum_i \eta_i \langle s_i \rangle_{\rm Data} - \log(Z).
\end{equation}
Different learning methods have been set up and proven to work efficiently, in particular the
contrastive divergence (CD) algorithm from Hinton~\cite{Hinton_CD} and more recently TAP based learning~\cite{TAP_train}.
They all correspond to expressing the gradient ascent on the likelihood as
\begin{align}
\Delta W_{ij} &= \gamma \left( \langle s_i \sigma_j p(\sigma_j|\bm{s}) \rangle_{\rm Data} - \langle s_i \sigma_j \rangle_{p_{\rm RBM}} \right)\label{eq:cd1}\\[0.2cm]
\Delta \eta_i &= \gamma \left( \langle s_i \rangle_{p_{\rm RBM}} - \langle s_i \rangle_{\rm Data} \right)\label{eq:cd_eta}\\[0.2cm]
\Delta \theta_j &= \gamma \left( \langle \sigma_j \rangle_{p_{\rm RBM}} -\langle \sigma_j p(\sigma_j|\bm{s}) \rangle_{\rm Data} \right)\label{eq:cd_theta}
\end{align}
where $\gamma$ is the learning rate. The main problem are the $\langle \cdots \rangle_{p_{\rm RBM}}$ terms  on the right hand side of~(\ref{eq:cd1}-\ref{eq:cd_theta}).
These are not tractable and the various methods basically differ in their way of estimating those terms (Monte-Carlo Markov chains, naive mean-field, TAP\ldots).
For an efficient learning the $\langle \cdots \rangle_{\rm Data}$ terms must also be approximated
by making use of random mini-batches of data at each step.

\section{Static thermodynamical properties of an RBM}\label{sec:thermo}

\subsection{Statistical ensemble of RBMs}\label{sec:stat-ensemble}
When analyzing the thermodynamical properties of RBMs, it is common to assume that the weights $W_{ij}$
are i.i.d. random variables, like for example in~\cite{TuMo,Barra,HHuang}. This
generally leads to a Marchenko-Pastur (MP) distribution~\cite{MP} of the singular values of $W$,
which is unrealistic.

In order to clarify our notation, let us recall the definition of the singular value decomposition (SVD).
As a generalization of eigenmodes decomposition to rectangular matrices, the SVD for a RBM is given by
\begin{equation}
\mathbf{W} = \mathbf{U \Sigma} \mathbf{V}^T
\end{equation}
where \(\mathbf{U}\) is an orthogonal \(N_v \times N_h\)  matrix whose columns are the left singular vectors \(\mathbf{u}^{\alpha} \), \(\mathbf{V}\) is an orthogonal \(N_h \times N_h\) matrix whose columns are the right singular vectors \( \mathbf{v}^{\alpha} \) and \( \mathbf{\Sigma} \) is a diagonal matrix whose elements are the singular values \(w_{\alpha}\). The separation into left and right singular vectors is due to the rectangular nature of the decomposed matrix, and the similarity with eigenmodes decomposition
is revealed by the following SVD equations

\begin{align}
\mathbf{W} \mathbf{v}^{\alpha} & = w_{\alpha} \mathbf{u}^{\alpha} \nonumber \\
\mathbf{W}^T \mathbf{u}^{\alpha} & = w_{\alpha} \mathbf{v}^{\alpha} \nonumber
\end{align}

In~\cite{DeFiFu} it is argued that the MP distribution of SVD modes actually corresponds to the noise of the weight matrix,
while the information content of the RBM is better expressed by the presence of SVD modes outside of this bulk.
This leads us to write the weight matrix as
\begin{equation}\label{eq:wsvd}
	W_{ij} = \sum_{\alpha=1}^{K} w_\alpha u_i^{\alpha} v_j^{\alpha} + r_{ij}
\end{equation}
where the $w_\alpha = O(1)$ are isolated singular values (describing a rank $K$ matrix), the $\bm{u}^\alpha$ and $\bm{v}^\alpha$ are the dominant eigenvectors of the SVD 
decomposition and the $r_{ij}={\cal N}(0,\sigma^2/L)$ are i.i.d. terms corresponding to noise, with $L=\sqrt{N_h N_v}$.
The $\{u^\alpha\}$ and $\{v^\alpha\}$  are two sets of respectively $N_v$ and $N_h$-dimensional orthonormal 
vectors, which means that their components are respectively 
$O(1/\sqrt{N_v})$ and $O(1/\sqrt{N_h})$, and $K\le N_v,N_h$. We assume $N_h<N_v$ to be the rank of $W$ and $w_\alpha >0$ and $O(1)$ for all $\alpha$. 
Note that in the limit $N_v\to\infty$ and $N_h\to\infty$ with $\kappa \egaldef N_h/N_v$ fixed and $K/L\to 0$, 
$WW^T$ has a spectrum density $\rho(\lambda)$ composed of a Marchenko-Pastur bulk of eigenvalues
and of set of discrete modes:
\[
\rho(\lambda) = \frac{L}{2\pi\sigma^2}\frac{\sqrt{(\lambda^+-\lambda)(\lambda-\lambda^-)}}{\kappa\lambda}
\ind{\lambda\in[\lambda^-,\lambda^+]}
+\sum_{\alpha=1}^K\delta(\lambda-w_\alpha^2),
\]
with
\[
\lambda^{\pm} \egaldef \sigma^2\bigl(\kappa^{\frac{1}{4}}\pm\kappa^{-\frac{1}{4}}\bigr)^2.
\]
The interpretation for the noise term $r_{ij}$ is given by the presence of an extensive number of modes at the bottom of the spectrum, along which the variables won't be able
to condense but that still contribute to the fluctuations.
In the present form our model of RBM is similar
to the Hopfield model and recent generalizations~\cite{Mezard},
the patterns being represented by the SVD modes outside of the bulk. The main difference, in addition to the bipartite structure of the graph,
is the non-degeneracy of the singular values $w_\alpha$.
The choice made here is to consider $K$ finite, giving $W_{ij} = O(1/N)$ which means that the thresholds $\theta_j$ (having the meaning of
feature detectors) should be $O(1)$ because feature $j$ is detected when an extensive number of spins $S_i$ is aligned with $W_{ij}$.
In addition, this allows us to assume simple distributions for the components of
$\bm{u}^\alpha$ and $\bm{v}^\alpha$ (for instance, considering them i.i.d.).
Altogether, this defines the statistical ensemble of RBM to which we restrict our analysis of the learning procedure.

Another approach would be to consider $K=N_h$ extensive, thereby assuming that all modes can potentially condense even though they are associated
to dominated singular values. In that case, the separation between the condensed modes and the rest should be made when order parameters are introduced
and the noise would then correspond to uncondensed modes. If the number of condensed modes is assumed to be extensive, then
we should instead consider an average over the orthogonal group which would lead to a slightly
different mean-field theory~\cite{PaPo,OpWi}.

\subsection{Replica symmetric Mean-field equation}\label{sec:RS}

Our analysis in the thermodynamic limit follows classical treatments using replicas,
like~\cite{AmGuSo1,AmGuSo3} for the Hopfield model or~\cite{Barra} for bipartite models.
The starting point is to express the average over $u,v$ and $r_{ij}$ of the log partition function $Z$ in~(\ref{eq_proba_rbm}) with the help of the replica trick:
\begin{equation*}
\EE_{u,v,r}[\log(Z)] = \lim_{p\to 0}\frac{d}{dp}\EE_{u,v,r}[Z^p].\\[0.2cm]
\end{equation*}
First the average over $r_{ij}$ yields
\[
\exp\Bigl[\frac{\sigma^2}{2L}\Bigl(\sum_a s_i^a\sigma_j^a\Bigr)^2\Bigr] =
\exp\Bigl[\frac{\sigma^2}{2L}\Bigl(p+\sum_{a\ne b} s_i^as_i^b\sigma_j^a\sigma_j^b\Bigr)\Bigr].
\]
After this averaging, 4 sets of order parameters $\{(m_\alpha^a,\bar m_\alpha^a),a=1,\ldots p,\alpha=1,\ldots K\}$
and $\{(Q_{ab},\bar Q_{ab}),a,b=1,\ldots p, a\ne b\}$  are introduced with the help of two distinct
Hubbard-Stratonovich transformations.
The first one corresponds to
\begin{align*}
\exp\Bigl[\frac{\sigma^2}{2L}\Bigl(\sum_{i,j,a\ne b} s_i^as_i^b\sigma_j^a\sigma_j^b\Bigr)\Bigr]
&= \int\prod_{a\ne b}\frac{dQ_{ab}d\bar Q_{ab}}{2\pi}\\[0.2cm]
&\times
\exp\Bigl[-\frac{L\sigma^2}{2}\sum_{a\ne b}\Bigl(Q_{ab}\bar Q_{ab} - \frac{Q_{ab}}{N_v}\sum_i s_i^a s_i^b
-\frac{\bar Q_{ab}}{N_h}\sum_j\sigma_j^a\sigma_j^b\Bigr)\Bigr].
\end{align*}
The second one is aimed at extracting magnetization's contributions correlated with the modes:
\begin{align*}
\exp\Bigl(L\sum_{\alpha}w_\alpha s_\alpha^a\sigma_\alpha^a\Bigr) & \propto
\int\prod_\alpha \frac{dm_\alpha^a d\bar m_\alpha^a}{2\pi} \\[0.2cm]
& \times \exp\Bigl(-L\sum_\alpha w_\alpha\bigl(m_\alpha^a\bar m_\alpha^a- m_\alpha^a s_\alpha^a
-\bar m_\alpha^a\sigma_\alpha^a \bigr) \Bigr),
\end{align*}
with
\begin{equation}
s_\alpha^a \egaldef \frac{1}{\sqrt{L}}\sum_i s_i u_i^\alpha\qquad\text{and}\qquad 
\sigma_\alpha^a \egaldef \frac{1}{\sqrt{L}}\sum_j\sigma_j^a v_j^\alpha,\label{eq:salpha}
\end{equation}
These variables represent the following quantities:
\begin{align*}
m_\alpha^a \sim  E_{u,v,r}\bigl(\langle \sigma_\alpha^a\rangle\bigr)\qquad
\bar m_\alpha^a \sim E_{u,v,r}\bigl(\langle s_\alpha^a\rangle\bigr)\\[0.2cm]
Q_{ab} \sim E_{u,v,r}\bigl(\langle \sigma_i^a \sigma_i^b\rangle\bigr)\qquad
\bar Q_{ab} \sim E_{u,v,r}\bigl(\langle s_j^a s_j^b\rangle\bigr),
\end{align*}
namely the correlations of the hidden [resp. visible] states with the left [resp. right]
singular vectors and the Edward-Anderson (EA) order parameters measuring the correlation
between replicas of hidden or visible states.
$\EE_{u}$ and  $\EE_{v}$ denote an average  w.r.t. the rescaled components $u\simeq \sqrt{N_v}u_i^\alpha$ and $v\simeq \sqrt{N_h}v_j^\alpha$ of the SVD modes.
The transformations involve pairs of complex integration variables because of the asymmetry introduced by the two-layers structure
in contrast to fully connected models.

We obtain the following representation:
\begin{align*}
\EE_{u,v,r}[Z^p] &= \int \prod_{a,\alpha}\frac{dm_\alpha^a d\bar m_\alpha^a}{2\pi}\prod_{a\ne b}\frac{dQ_{ab}d\bar Q_{ab}}{2\pi}\\[0.2cm]
&\times\exp\Bigl\{-L\Bigl(\sum_{a,\alpha}w_\alpha m_\alpha\bar m_\alpha+\frac{\sigma^2}{2}\sum_{a\ne b}Q_{ab}\bar Q_{ab}-\frac{1}{\sqrt{\kappa}}A[m,Q]
-\sqrt{\kappa}B[\bar m,\bar Q]\Bigr)\Bigr\}
\end{align*}
with $\kappa=N_h/N_v$ and
\begin{align}
A[m,Q] &\egaldef \log\Bigl[\sum_{S^a\in\{-1,1\}}\EE_u \Bigl(e^{\frac{\sqrt{\kappa}\sigma^2}{2}\sum_{a\ne b}Q_{ab}S^a S^b
+\kappa^{\frac{1}{4}}\sum_{a,\alpha}(w_\alpha m_\alpha^a -\eta_\alpha)u^\alpha S^a}\Bigr)\Bigr],\label{eq:Ap}\\[0.2cm]
B[\bar m,\bar Q] &\egaldef \log\Bigl[\sum_{S^a\in\{-1,1\}}
\EE_v \Bigl(e^{\frac{\sqrt{\kappa}\sigma^2}{2}\sum_{a\ne b}\bar Q_{ab}\sigma^a \sigma^b +\kappa^{-\frac{1}{4}}\sum_{a,\alpha}(w_\alpha \bar m_\alpha^a -\theta_\alpha)v^\alpha \sigma^a}\Bigr)\Bigr],\label{eq:Bp}\\[0.2cm]
\end{align}
with
\[
\theta_\alpha \egaldef \frac{1}{\sqrt{L}}\sum_j \theta_j v_j^\alpha = O(1).
\]
Since $\{v^\alpha\}$ is an incomplete basis we also need to take care of the potential residual transverse parts $\eta^{\bot}$ and $\theta^{\bot}$, such
that the following decompositions hold:
\begin{align}
\eta_i &= \eta_i^{\bot} +\sqrt{L}\sum_\alpha\eta_\alpha u_i^\alpha,\label{def:eta}\\[0.2cm]
\theta_j &= \theta_j^{\bot} +\sqrt{L}\sum_\alpha\theta_\alpha v_j^\alpha.\label{def:theta}
\end{align}
To keep things tractable, both $\eta^{\bot}$ and $\theta^{\bot}$ will be considered negligible in the sequel. 
Taking into account these components would lead to the addition of a random field to the effective RS field of the variables and 
eventually to a richer set of saddle point solutions. 
Note that the order of magnitude of $\eta_\alpha$ and $\theta_\alpha$ is at this stage an assumption. 
If $\eta_i$ and $u_i^\alpha$ (or $\theta_j$ and $v_j^\alpha$) were uncorrelated they would scale as $1/\sqrt{L}$. 
Moreover, regarding the ensemble average, we will consider $\eta_\alpha$ and $\theta_\alpha$ fixed in the sequel.

The thermodynamic properties are obtained by first making a saddle point approximation possible by letting first $L\to\infty$ and taking the limit $p\to 0$ afterwards.
We restrict here the discussion to RS saddle points~\cite{MePaVi}. The breakdown of RS can actually be determined
by computing the so-called AT line~\cite{AlTh} (see Appendix~\ref{app:AT}).
At this point we assume a non-broken replica symmetry.
The set $\{Q_{ab},\bar Q_{ab}\}$ reduces then to a pair $(q,\bar q)$  of spin glass parameters,
i.e. $Q_{ab} = q$ and $\bar Q_{ab} = \bar q$ for all $a\ne b$,
while  quenched magnetizations on the SVD directions are now represented by
$\{(m_\alpha,\bar m_\alpha),\alpha=1,\ldots K\}$.

Taking the limit $p\to 0$ yields the following limit for the free energy:
\begin{align}
f[m,\bar m,q,\bar q] & = \sum_\alpha w_\alpha m_\alpha\bar m_\alpha - \frac{\sigma^2}{2}q\bar q +\frac{\sigma^2}{2}(q+\bar q) \nonumber\\[0.2cm]
&-\frac{1}{\sqrt{\kappa}} \EE_{u,x}\Bigl[\log2\cosh\bigl(h(x,u)\bigr)\Bigr]
-\sqrt{\kappa} \EE_{v,x}\Bigl[\log2\cosh\bigl(\bar h(x,v)\bigr)\Bigr].\label{eq:freen}
\end{align}
Assuming a replica-symmetric phase, the saddle-point equations are given by
\begin{align}
m_\alpha &= \kappa^{\frac{1}{4}}\EE_{v,x}\Bigl[v^\alpha\tanh\bigl(\bar h(x,v)\bigr)\Bigr],\qquad\qquad
q = \EE_{v,x}\Bigr[\tanh^2\bigl(\bar h(x,v)\bigr)\Bigr]\label{eq:mf1}\\[0.2cm]
\bar m_\alpha &= \kappa^{-\frac{1}{4}}\EE_{u,x}\Bigl[u^\alpha\tanh\bigl(h(x,u)\bigr)\Bigr]
,\qquad\qquad \bar q = \EE_{u,x}\Bigl[\tanh^2\bigl(h(x,u)\bigr)\Bigr]\label{eq:mf2}
\end{align}
where
\begin{align*}
h(x,u) &\egaldef  \kappa^{\frac{1}{4}}\bigl(\sigma\sqrt{q}x+\sum_\gamma (w_\gamma m_\gamma-\eta_\gamma) u^\gamma\bigr)\\[0.2cm]
\bar h(x,v) &\egaldef \kappa^{-\frac{1}{4}}\bigl(\sigma\sqrt{\bar q}x+\sum_\gamma (w_\gamma\bar m_\gamma-\theta_\gamma) v^\gamma\bigr),
\end{align*}
and $\kappa=N_h/N_v$, with $\EE_{u,x}$ and  $\EE_{v,x}$ denoting an average over the Gaussian variable 
$x={\cal N}(0,1)$ and the rescaled components $u\sim \sqrt{N_v}u_i^\alpha$ and $v\sim \sqrt{N_h}v_j^\alpha$ 
of the SVD modes. We note that the equations are symmetric under the exchange $\kappa\to\kappa^{-1}$, simultaneously with 
$m\leftrightarrow\bar m$, $q\leftrightarrow\bar q$ and $\eta\leftrightarrow\theta$, given that $u$ and $v$ have the same distribution.
In addition, for independently distributed $u_i^\alpha$ and $v_j^\alpha$ and vanishing fields ($\eta=\theta=0$),
solutions corresponding to non-degenerate magnetizations have symmetric counterparts: each pair of non-vanishing magnetizations
can be negated independently as $(m_\alpha,\bar m_\alpha)\to (-m_\alpha,-\bar m_\alpha)$, generating new solutions. So to one solution presenting $n$ condensed modes, there correspond $2^n$ distinct solutions.

\subsection{Phase Diagram}\label{sec:phasediag}

The fixed point equations (\ref{eq:mf1},~\ref{eq:mf2}) can be solved numerically to tell us how the variables condensate on the SVD modes within each equilibrium state of the distribution
and whether a spin-glass or a ferromagnetic phase is present. The important point here is that
with $K$ finite and a non-degenerate spectrum the mode with highest singular value dominates the ferromagnetic phase.

In absence of bias ($\eta = \theta=0$) and once $1/\sigma$ is interpreted as
temperature and $w_\alpha/\sigma$ as ferromagnetic couplings, we get a phase diagram similar to that of the Sherrington-Kirkpatrick (SK) model
with three distinct phases (see Figure~\ref{fig:phasediag})
\begin{itemize}
\item a paramagnetic phase ($q=\bar q= m_\alpha = \bar m_\alpha = 0$) (P),
\item a ferromagnetic phase ($q,\bar q,m_\alpha,\bar m_\alpha \ne 0$) (F),
\item a spin glass phase ($q,\bar q\ne 0$; $m_\alpha = \bar m_\alpha = 0$) (SG).
\end{itemize}
In general, the lines separating the different phases correspond to second order phase transitions and can be obtained
by a stability analysis of the Hessian of the free energy. They are related to unstable modes of the linearized mean-field equations and correspond to an eigenvalue of the Hessian
becoming negative.

The (SG-P) line is obtained by looking at the Hessian in the $(q,\bar q)$ sector:
\[
H_{q\bar q} \eg_{m=0\atop q=0} -\frac{1}{2}\left[
\begin{matrix}
\sigma^2 & \frac{\sigma^4}{\sqrt{\kappa}} \\
\sqrt{\kappa}\sigma^4 & \sigma^2
\end{matrix}
\right]
\]
from what results that the spin glass phase develops when $\sigma\ge 1$\footnote{Note that in~\cite{Barra} a dependence $\sqrt{\kappa(1-\kappa)}$ 
$\left( \sqrt{\alpha(1-\alpha)} \text{in their notation} \right)$ is found. This dependence is  hidden in our definition of $\sigma^2$ giving $L=\sqrt{N_v N_h}$ times the variance of 
$r_{ij}$ instead of $N_v+N_h$ as in their case.}.
This transition line is understood tacking directly into account the spectral properties of the weight matrix. 
Classically, this is done with the help of the linearized TAP equations and exploiting the Marchenko-Pastur distribution~\cite{MePaVi}. 
In our context, the linearized TAP equations read
\[
\left[
\begin{matrix}
\mu\\
\nu
\end{matrix}\right] =
\left[
\begin{matrix}
-\sqrt{\kappa}\sigma^2 & W^T\\
W &-\frac{\sigma^2}{\sqrt{\kappa}}
\end{matrix}
\right]
\left[
\begin{matrix}
\mu\\
\nu
\end{matrix}\right]
\]
given the variance $\sigma^2/L$ of the weights in absence of dominant modes. Then we can show that the paramagnetic phase becomes unstable when the highest eigenvalue of the matrix on the rhs
is equal to $1$: if $\lambda$ is a singular value of $W$, the corresponding eigenvalues $\Lambda^{\pm}$ verify the relation
\[
\bigl(\frac{\Lambda^{\pm}}{\sqrt{\kappa}}\pm\sigma^2\bigr)\bigl(\sqrt{\kappa}\Lambda^{\pm}\pm\sigma^2\bigr) = \lambda^2.
\]
from which it is clear that the largest eigenvalue $\Lambda_{max}$ corresponds to the largest singular value $\lambda_{max}$. Owing to
the Marchenko-Pastur distribution $\lambda_{max}=\sigma^2(\sqrt{\kappa}+1)(1+1/\sqrt{\kappa})$ so $\Lambda_{max}$ verifies
\[
\bigl(\frac{\Lambda_{max}}{\sqrt{\kappa}}+\sigma^2\bigr)\bigl(\sqrt{\kappa}\Lambda_{max}+\sigma^2\bigr) = \sigma^2(\sqrt{\kappa}+1)\bigl(\frac{1}{\sqrt{\kappa}}+1\bigr).
\]
$\Lambda_{max} = 1$ is readily obtained for $\sigma^2=1$.

For the (F-SG) frontier we can look at the sector $(m_\alpha,\bar m_\alpha)$ corresponding to the emergence of a single mode $\alpha$ (written in the spin-glass phase):
\begin{align*}
H_{\alpha\alpha} &= \left[
\begin{matrix}
w_\alpha & w_\alpha^2\EE_{v,x}\Bigl[(v^\alpha)^2\sech^2\bigl(\bar h(x,v))\bigr)\Bigr] \\
w_\alpha^2\EE_{u,x}\Bigl[(u^\alpha)^2\sech^2\bigl(h(x,u)\Bigr)\Bigr] & w_\alpha
\end{matrix}
\right] \\[0.2cm]
&\eg_{m_\alpha=0}
\left[
\begin{matrix}
w_\alpha & w_\alpha^2(1-q) \\
w_\alpha^2(1-\bar q) & w_\alpha
\end{matrix}\right]
\end{align*}
From this it is clear that
the first mode to become unstable is the mode $\alpha$ with highest singular value $w_\alpha$ and this occurs when $q$ and $\bar q$, solutions of (\ref{eq:mf1},\ref{eq:mf2}),
verify
\[
(1-q)(1-\bar q) w_\alpha^2 = 1.
\]
As for the SK model, this line appears to be well below the de Almeida-Thouless (AT) line, which is the line above which the RS solution is stable (see Figure~\ref{fig:phasediag}, and Appendix~\ref{app:AT} for the computation of the AT line). This means that in principle
a replica symmetry breaking treatment would be necessary to properly separate the two phases. However, we will leave aside this point as we are mainly interested in the practical aspects, namely
the ability of the RBM to learn arbitrary data, and so we are mostly concerned with the ferromagnetic phase above the AT line.

For the (P-F) line we consider the same sector of the Hessian but now written in the paramagnetic phase, 
i.e. setting $q=0$ in the above equation, and this simply yields the emergence of the single mode $\alpha$ for $w_\alpha=1$.

Note that all of this is independent on how the statistical average over $u$ and $v$ is performed.
Instead, as we shall see later on, the way of averaging influences the nature of the ferromagnetic phase.

\begin{figure}[ht]
\centering{
\includegraphics[width=\textwidth]{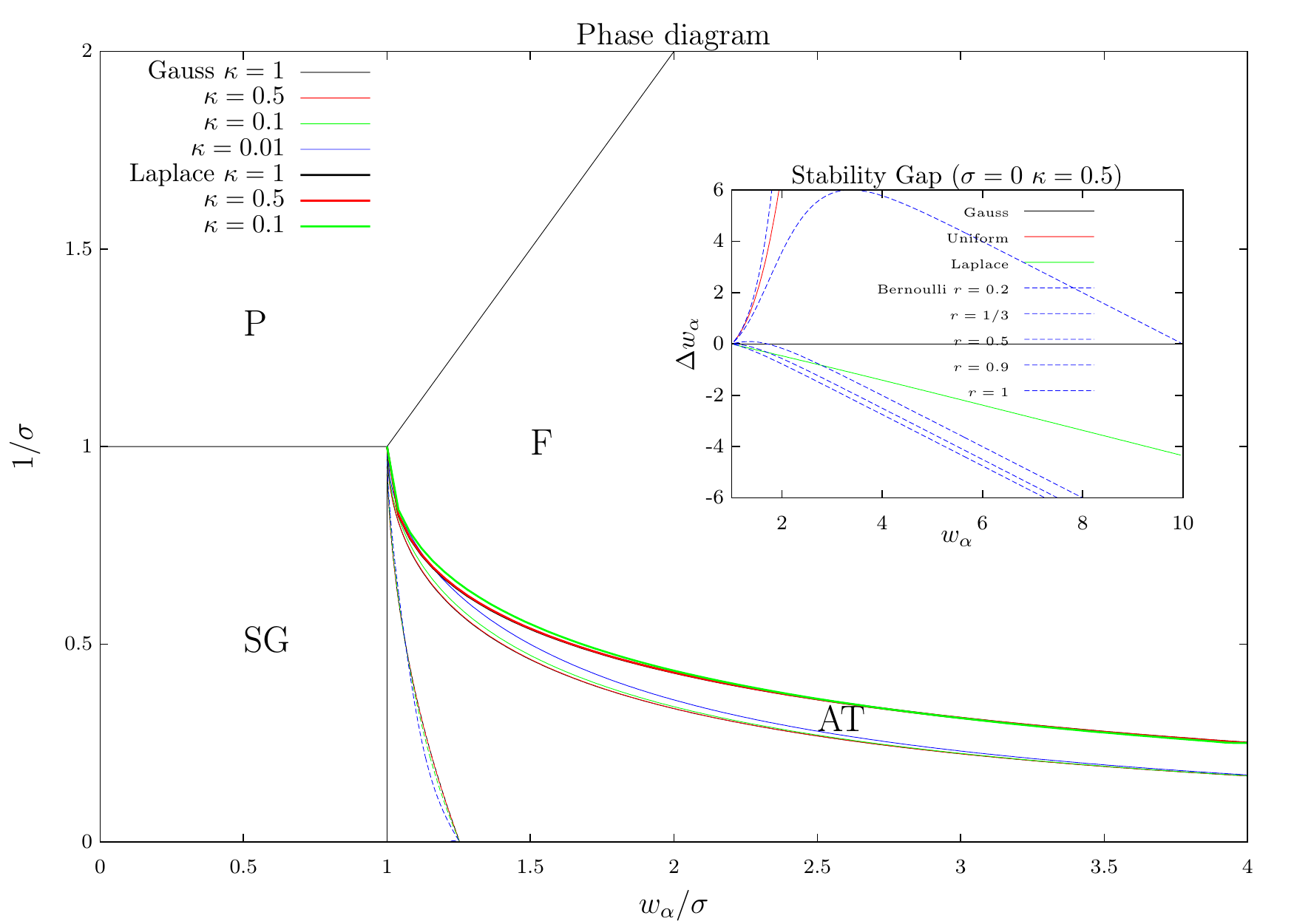}}
\caption{Phase diagram in absence of bias and with a finite number of modes, with Gaussian and Laplace distributions for $u$ and $v$.
The dotted line separates the spin glass phase from the ferromagnetic phase under the RS hypothesis.
The RS phase is unstable below the AT line. The influence of $\kappa$ on the AT and SG-F lines is shown.
In all cases, the hypothetical SG-F line lies well inside the broken RS phase. Inset: high temperature ($\sigma=0$) 
stability gap $\Delta w_\alpha$ corresponding to a fixed point associated to a mode $\beta$, expressed as a function of $w_\alpha$ and considering various distributions.
}\label{fig:phasediag}
\end{figure}

Regarding the stability of the RS solution, the computation of the AT line reported in Appendix~\ref{app:AT} is similar to the classical one
made for the SK model, though slightly more involved. In fact we were not able to fully characterize, in replica space, all the possible
instabilities of the Hessian which would potentially lead to a breakdown of the replica symmetry.
At least the one responsible for  the ordinary SK model RS breakdown has a counterpart in the bipartite
case that gives a necessary condition for the stability of the RS solution:
\[
\frac{1}{\sigma^2} > \sqrt{\EE_{x,u}\Bigl(\sech^4\bigl(h(x,u)\bigr)\Bigr)\EE_{x,v}\Bigl(\sech^4\bigl(\bar h(x,v)\bigr)\Bigr)},
\]
For $\kappa=1$ the terms below the radical 
become identical and the condition reduces to the one of the SK model, except for the $u$ averages which are not
present in the SK model. In Figure~\ref{fig:phasediag}, is shown the influence on the phase diagram of the value of $\kappa$ and of the type of average
made on $u$ and $v$.

\subsection{Nature of the Ferromagnetic phase}\label{sec:ferro_phase}

Some subtleties arise when considering various ways of averaging over the components of the singular vectors.
In~\cite{Agliari,TuMo} is emphasized the importance for networks to be able to reproduce compositional states
structured by combination of hidden variables. In our representation, we don't have direct access to this property
but, in some sense, to the dual one, which is given by states corresponding to combinations of modes.
Their presence and their structure are rather sensitive to the way the average over $u$ and $v$ is performed.
In this respect the case in which $\bm{u}^\alpha$ and $\bm{v}^\alpha$ have i.i.d. Gaussian components is very special:
all fixed points associated to dominated modes can be shown to be unstable and fixed points associated to combinations of modes are not allowed.
To see this, first notice that in such a case the magnetization's part of the saddle point equations (\ref{eq:mf1},\ref{eq:mf2}) read
\begin{align}
m_\alpha &= (w_\alpha\bar m_\alpha-\theta_\alpha)
(1-q)\label{eq:gmf1} \\[0.2cm]
\bar m_\alpha &= (w_\alpha m_\alpha-\eta_\alpha) (1-\bar q).\label{eq:gmf2}
\end{align}
Since the role of the bias is mainly to introduce some asymmetry between otherwise degenerated fixed points obtained by sign reversal of at least 
one pair $(m_\alpha,\bar m_\alpha)$, let us analyze the situation without fields, i.e. by setting $\eta=\theta=0$. We immediately see that as long as the 
singular values are non degenerate, only one single mode may condense at a time. Indeed if mode $\alpha$ condenses we necessarily have 
\[
w_\alpha^2(1-q)(1-\bar q) = 1,
\]
and this can be verified only by one mode at a time. Looking at the stability of the fixed points, we see that only the fixed point associated to
the largest singular value is actually stable (details reported after the introduction of lemma \ref{lem:p_star_var}).

For other distributions like uniform Bernoulli or Laplace, instead, stable fixed points associated to many different single modes or
combinations of modes can exist and contribute to the thermodynamics. In order to
analyze this question in more general terms we first rewrite the mean-field equations in a convenient way which require some preliminary remarks.
We restrict the discussion to i.i.d. variables so that we can consider single variable distributions. Joint distributions will be distinguished 
from single variable distributions by the use of bold: $\bu=\{u^\alpha,\alpha=1,\ldots,K\}$, $K$ being the (finite) number of modes
susceptible of condensing.

Given the distribution $p$ and assuming it to be even, we define a related distribution $p^\star$ attached to mode $\alpha$:
\begin{equation}\label{def:pualpha}
p^\star(u) \egaldef -\int_{-\infty}^{u} x p(x)dx 
= \int_{\vert u\vert}^\infty x p(x)dx,
\end{equation}
This distribution has some useful properties.
\begin{lem}
Given that  $p$ is centered with unit variance and kurtosis $\kappa_u$, 
$p^\star$ is a centered probability distribution with variance 
\[
\int_{-\infty}^\infty u^2 p^\star(u)du  = \frac{\kappa_u}{3}.
\]  
\label{lem:p_star_var}
\end{lem}
\begin{proof}
Consider the moments of $p^\star$. For $n$ odd they vanish while for $n$ even they read:
\begin{align*}
\int_{-\infty}^{+\infty} u^n p^\star(u)du  &= 2\int_0^\infty u^n p^\star(u)du \\[0.2cm] 
&= 2\int_0^\infty du u^n \int_{u}^\infty x p(x)dx \\[0.2cm]
&= 2\int_0^\infty x p(x)dx \int_0^{x} u^n du\\[0.2cm]
&= \frac{1}{n+1}\int_{-\infty}^\infty x^{n+2}p(x)dx , 
\end{align*}
	i.e. the $n_{th}$ even moments of $p^\star$ relate to moments of order $n+2$ of $p$. The lemma then follows from the fact that $p$ has unit variance.
\end{proof}
In this respect, the Gaussian averaging is special because we have $\kappa_u=3$ and $p^\star = p$. Then the mean-field equations (\ref{eq:mf1},\ref{eq:mf2}) corresponding to the magnetizations can be rewritten 
in a form similar to (\ref{eq:gmf1},\ref{eq:gmf2}) by introducing the variables $q_\alpha$ and $\bar q_\alpha$:
\begin{align}
m_\alpha &= (w_\alpha\bar m_\alpha - \theta_\alpha)(1-q_\alpha),\label{eq:mf1_alpha}\\[0.2cm]
\bar m_\alpha &= (w_\alpha m_\alpha - \eta_\alpha)(1-\bar q_\alpha),\label{eq:mf2_alpha}
\end{align}
with
\begin{align}
q_\alpha &= \int dx \frac{e^{-x^2/2}}{\sqrt{2\pi}} d\bv p_\alpha(\bv) \tanh^2
\Bigl(\kappa^{-\frac{1}{4}}\bigl(\sigma\sqrt{\bar q}x+\sum_\gamma (w_\gamma\bar m_\gamma-\theta_\gamma) v^\gamma\bigr)\Bigr),\label{eq:mf3_alpha}\\[0.2cm]
\bar q_\alpha &= \int dx \frac{e^{-x^2/2}}{\sqrt{2\pi}} d\bu p_\alpha(\bu) \tanh^2
\Bigl(\kappa^{\frac{1}{4}}\bigl(\sigma\sqrt{q}x+\sum_\gamma (w_\gamma m_\gamma-\eta_\gamma) u^\gamma\bigr)\Bigr),\label{eq:mf4_alpha}
\end{align}
where
\[
p_\alpha(\bu) \egaldef p^\star(u^\alpha)\prod_{\beta\ne\alpha}p(u^\beta).
\]
This rewriting will prove very useful also in the next section when analyzing the learning dynamics. 

Let us now assume, in absence of bias, a non-degenerate fixed point associated to some given mode $\beta$ with 
finite $(m_\beta,\bar m_\beta)$ and $m_\alpha=\bar m_\alpha = 0,\forall \alpha\ne \beta$. The fixed point equation imposes the relation
\begin{equation}\label{eq:wqbq}
w_\beta = \frac{1}{\sqrt{(1-q_\beta)(1-\bar q_\beta)}} \egaldef w(q_\beta,\bar q_\beta).
\end{equation}
The stability of such a fixed point with respect to any other mode $\alpha$ is related to the positive definiteness of 
the following block of the Hessian 
\[
H_{\alpha\alpha} = \left[
\begin{matrix}
w_\alpha & w_\alpha^2\EE_{v,x}\Bigl[(v^\alpha)^2\sech^2\bigl(\bar h(x,v)\bigr)\Bigr] \\[0.2cm]
w_\alpha^2\EE_{u,x}\Bigl[(u^\alpha)^2\sech^2\bigl(h(x,u)\bigr)\Bigr] & w_\alpha
\end{matrix}
\right]
\]
with, in the present case
\[
h(x,u) = \kappa^{\frac{1}{4}}\bigl(\sigma\sqrt{q}x+w_\beta\bar m_\beta u^\beta\bigr)\qquad\text{and}\qquad
\bar h(x,v) = \kappa^{-\frac{1}{4}}\bigl(\sigma\sqrt{\bar q}x+w_\beta\bar m_\beta v^\beta\bigr),
\]
This reduces to
\[
H_{\alpha\alpha} =
\left[
\begin{matrix}
w_\alpha & w_\alpha^2(1-q) \\[0.2cm]
w_\alpha^2(1-\bar q) & w_\alpha
\end{matrix}\right].
\]
Therefore for the Gaussian averaging case, since $q_\beta=q$, $\bar q_\beta = \bar q$ and given (\ref{eq:wqbq}),
we necessarily have
\[
1 - (1-q)(1-\bar q) w_\alpha^2 = 1-\frac{w_\alpha^2}{w_\beta^2} < 0\qquad \text{for}\qquad w_\alpha > w_\beta,
\]
i.e. the Hessian has negative eigenvalues.
This means that if the mode $\beta$ is dominated by another mode $\alpha$, the magnetization 
$(m_\alpha,\bar m_\alpha)$ will develop until $(1-q)(1-\bar q) w_\alpha^2=1$, while $m_\beta$ will vanish.

For the general case of i.i.d. variables, assuming $u^\alpha$  and $v^\alpha$ obey the same distribution
$p$, let $F$ and $F_\alpha$ be the cumulative distributions associated respectively
to $p$ and $p_\alpha$
\begin{align*}
F(u) &\egaldef \int_{-\infty}^{u}p(x)dx \\[0.2cm]
F_\alpha(u) &\egaldef \int d\bu\ \theta(u-u^\alpha)p_\alpha(\bu)dx = -\int_{\infty}^{u}du^\alpha\int_{-\infty}^{u^\alpha} x p(x)dx.
\end{align*}
Given the values of $(q,\bar q)$ obtained from the fixed point associated to mode $\beta$,
we have the following property:
\begin{prop}
If
\begin{align*}
(i)\ F_\beta(u) &< F(u),\qquad\forall u\in{\mathbb R^+}\qquad \text{then}\qquad q_\beta > q\qquad \text{and}\qquad \bar q_\beta> \bar q,\\[0.2cm]
(ii)\ F_\beta(u) &> F(u),\qquad\forall u\in{\mathbb R^+}\qquad \text{then}\qquad q_\beta < q\qquad \text{and}\qquad \bar q_\beta< \bar q,
\end{align*}
which in turn implies
\[
w(q,\bar q)  < w_\beta\  (i) \qquad \text{and}\qquad w(q,\bar q) > w_\beta\ (ii)
\]
with
\[
w(q,\bar q) \egaldef \frac{1}{\sqrt{(1-q)(1-\bar q)}}.
\]
\end{prop}
\begin{proof}
This is obtained by straightforward by parts integration respectively over $u$ and $v$ in equations~(\ref{eq:mf1},\ref{eq:mf2}),
relative to magnetizations.
\end{proof}
In other words if  $F_\beta$ dominates  $F$ on ${\mathbb R}^+$
then there is a positive stability gap defined as
\begin{equation}
\Delta w_\beta\egaldef w(q,\bar q)-w_\beta \label{eq:gap}
\end{equation}
such that
there is a non-empty range
for higher values of $w_\alpha\in [w_\beta,w(q,\bar q)[$ for which
the fixed point associated to mode $\beta$ corresponds to a local minimum of the free energy.
Note that property (i) [resp. (ii)] is analogous (in the sense that it implies it) to $p_\beta$ having a larger [resp. smaller] variance than $p$, i.e. $\kappa_u >3$
[resp. $\kappa_u <3$].
Therefore distributions $p$ with negative relative kurtosis ($\kappa_u-3$) will tend to favor the presence of metastable states, while the situation will tend to
be more complex for probabilities with positive relative kurtosis. Indeed, in the latter case the fixed point associated to the highest
mode $\alpha_{max}$ might not correspond to a stable state if lower modes in the range $[w(q,\bar q),w_{\alpha_{max}}[$ are present,
and fixed points associated to combinations of modes have to be considered. Note that in contrary with the Gaussian case,
this can happen because $q_\alpha$ is different for each mode and therefore more flexibility is offered by equations~(\ref{eq:mf1_alpha},\ref{eq:mf2_alpha})
than from equations (\ref{eq:gmf1},\ref{eq:gmf2}).

Let us give some examples. Denote by $\gamma_u \egaldef \kappa_u-3$ the relative kurtosis.
As already said the Gaussian distribution is a special case with $\gamma_u=0$. In addition, for instance for
$p$ corresponding to Bernoulli, Uniform or Laplace, we have the following properties illustrated in the inset of Figure~\ref{fig:phasediag}:
\begin{itemize}
\item Bernoulli ($\gamma_u = -2$):
\begin{align*}
p(u) &= \frac{1}{2}\bigl(\delta(u+1)+\delta(u-1)\bigr),\qquad  F(u) = \frac{1}{2}\bigl(\theta(u+1)+\theta(u-1)\bigr) \\
p_\alpha(u) &= \frac{1}{2}\theta(1-u^2),\qquad
F_\alpha(u) = \frac{1}{2}\theta(1-u^2)(u+1)+\theta(u-1)
\end{align*}
then $F_\alpha(u) > F(u)$ for $u>0$, yielding a positive stability gap.
\item Uniform ($\gamma_u = -6/5$):
\begin{align*}
p(u) &= \frac{1}{2\sqrt{3}}\theta(3-u^2),\qquad  F(u) = \frac{1}{2\sqrt{3}}\theta(3-u^2)(u+\sqrt{3})+\theta(u-\sqrt{3}) \\
p_\alpha(u) &= \frac{1}{4\sqrt{3}}\theta(3-u^2)(3-u^2),\qquad
F_\alpha(u) = \frac{1}{4\sqrt{3}}\theta(3-u^2)(3u-\frac{u^3}{3}+2\sqrt{3})+\theta(u-\sqrt{3}).
\end{align*}
It can be verified that $F_\alpha(u)> F(u)$ for $u>0$, yielding again a positive stability gap.
\item Laplace ($\gamma_u = 3$):
\begin{align*}
p(u) &= \frac{1}{\sqrt{2}}e^{-\sqrt{2}\vert u\vert},\qquad  F(u) = \frac{1}{2}+\frac{u}{2\vert u\vert}\bigl(1-e^{-\sqrt{2}\vert u\vert}\bigr)\\
p_\alpha(u) &= \frac{1}{2}\bigl(\vert u\vert+\frac{1}{\sqrt{2}}\bigr)e^{-\sqrt{2}\vert u\vert},\qquad
F_\alpha(u) = F(u)-\frac{u}{2\sqrt{2}}e^{-\sqrt{2}\vert u\vert}.
\end{align*}
Here we have $F_\alpha(u) < F(u)$ for $u>0$, yielding a negative stability gap.
\end{itemize}
These three examples fall either in condition (i) or (ii), with a stability gap $\Delta w_\beta$ that
is either always positive or always negative, independently of $w_\beta$.
We can also provide examples for which the stability condition may vary with $w_\beta$. Consider for instance a sparse Bernoulli
distribution, with $r\in[0,1]$ a sparsity parameter:
\[
p(u) = \frac{r}{2}\bigl(\delta(u+\frac{1}{\sqrt{r}})+\delta(u-\frac{1}{\sqrt{r}})\bigr)
+(1-r)\delta(u).
\]
The relative kurtosis is in this case
\[
\gamma_u(r) = \frac{1}{r}-3.
\]
Looking at $F(u)$ and $F_\alpha(u)$ it is seen that both conditions $(i)$ and $(ii)$ are not fulfilled,
except for $r=1$ which corresponds to the plain Bernoulli case. As we see in the inset of Figure~\ref{fig:phasediag},
for $r<1/3$ the stability gap is always negative, meaning that a unimodal ferromagnetic phase is not stable,
and it is replaced by a compositional ferromagnetic phase at all temperatures. Instead, for $r>1/3$ and at sufficiently high temperature
(low $w_\alpha$) the single mode fixed point dominate the ferromagnetic phase.

\paragraph{Laplace distribution:} let us look at the properties of the phase diagram in the case of singular vectors' components being Laplace i.i.d., case in which a negative stability gap is
expected and it may lead to a compositional phase. For this we need the expression for
a sum of Laplace variables to compute the averages involved in~(\ref{eq:mf1},\ref{eq:mf2}).
For this purpose, we define the following distributions:
\begin{align*}
f(s) &= \int\prod_\gamma du^\gamma\frac{\lambda_\gamma}{2} e^{-\lambda_\gamma  \vert u^\gamma\vert}\ \delta(s-\sum_\gamma u^\gamma),\\[0.2cm]
g_\alpha(s) &= \int du^\alpha \frac{\lambda_\alpha}{4}(\lambda_\alpha\vert u^\alpha\vert+1)e^{-\lambda_\alpha\vert u^\alpha\vert}
\prod_{\gamma\ne\alpha} du^\gamma\frac{\lambda_\gamma}{2} e^{-\lambda_\gamma  \vert u^\gamma\vert}\ \delta(s-\sum_\gamma u^\gamma).
\end{align*}
Their Laplace transform upon decomposing into partial fractions reads:
\[
\tilde f(\omega) = \prod_\gamma \frac{\lambda_\gamma^2}{\lambda_\gamma^2-\omega^2}
= \sum_\gamma C_\gamma \frac{\lambda_\gamma^2}{\lambda_\gamma^2-\omega^2}
\]
and
\begin{align*}
\tilde g_\alpha(\omega) &= \frac{\lambda_\alpha^2}{\lambda_\alpha^2-\omega^2}
\prod_\gamma \frac{\lambda_\gamma^2}{\lambda_\gamma^2-\omega^2}\\[0.2cm]
&= C_\alpha\frac{\lambda_\alpha^4}{(\lambda_\alpha^2-\omega^2)^2}+
\sum_{\gamma\ne\alpha}C_\gamma\frac{\lambda_\gamma^2\lambda_\alpha^2}{\lambda_\alpha^2-\lambda_\gamma^2}
\Bigl(\frac{1}{\lambda_\gamma^2-\omega^2}-\frac{1}{\lambda_\alpha^2-\omega^2}\Bigr).
\end{align*}
where
\[
C_\gamma \egaldef \prod_{\delta\ne\gamma}\frac{\lambda_\delta^2}{\lambda_\delta^2-\lambda_\gamma^2}.
\]
From these decompositions we immediately identify
\begin{align*}
f(s) &= \frac{1}{2}\sum_\gamma C_\gamma \lambda_\gamma e^{-\lambda_\gamma\vert s\vert},\\[0.2cm]
g_\alpha(s) &= \frac{\lambda_\alpha C_\alpha}{4}(\lambda_\alpha\vert s\vert+1) e^{-\lambda_\alpha\vert s\vert} +
\frac{1}{2}\sum_{\gamma\ne\alpha} C_\gamma\frac{\lambda_\gamma\lambda_\alpha}{\lambda_\alpha^2-\lambda_\gamma^2}
\bigl(\lambda_\alpha e^{-\lambda_\gamma\vert s\vert}-\lambda_\gamma e^{-\lambda_\alpha\vert s\vert}\bigr).
\end{align*}
This results in the following decomposition of the EA parameters:
\begin{align}
q &= \int dxds\frac{e^{-\sqrt{2}\vert s\vert-x^2/2}}{2\sqrt{\pi}}\sum_{\gamma}C_\gamma[\bar m]\tanh^2\bigl(\bar h_\gamma(x,s)\bigr)\label{eq:q_exp}\\[0.2cm]
q_\alpha &= \int dxds\frac{e^{-\sqrt{2}\vert s\vert-x^2/2}}{2\sqrt{\pi}}\Bigl[\frac{1}{\sqrt{2}}(\vert s\vert+\frac{1}{\sqrt{2}})
C_\alpha[\bar m] \tanh^2\bigl(\bar h_\alpha(x,s)\bigr)\\[0.2cm]
&+\sum_{\gamma\ne\alpha} C_\gamma[\bar m]
\frac{(w_\gamma\bar m_\gamma-\theta_\gamma)^2\tanh^2\bigl(\bar h_\gamma(x,s)\bigr)-
(w_\alpha\bar m_\alpha-\theta_\alpha)^2\tanh^2\bigl(\bar h_\alpha(x,s)\bigr)}
{(w_\gamma\bar m_\gamma-\theta_\gamma)^2-(w_\alpha\bar m_\alpha-\theta_\alpha)^2}
\Bigr]\label{eq:qalpha_exp}
\end{align}
with
\[
\bar h_\gamma(x,s) \egaldef  \kappa^{-\frac{1}{4}}\bigl(\sigma\sqrt{\bar q}x+(w_\gamma\bar m_\gamma-\theta_\gamma)s\bigr)
\]
and
\[
C_\gamma[\bar m] \egaldef \prod_{\delta\ne\gamma}\frac{(w_\gamma\bar m_\gamma-\theta_\gamma)^2}{(w_\gamma\bar m_\gamma-\theta_\gamma)^2
-(w_\delta\bar m_\delta-\theta_\delta)^2}.
\]
This allows for an efficient resolution of the mean-field equations~(\ref{eq:mf1},\ref{eq:mf2},\ref{eq:mf1_alpha},\ref{eq:mf2_alpha}), which let us observe the appearance of a purely compositional phase in the ferromagnetic domain
when the modes at the top of the spectrum get close enough. In order to characterize this phase,
we consider the stability gap $\Delta^{(n)}(w_\alpha)$ for which the range $[w_a-\Delta^{(n)}(w_\alpha),w_a]$
lies below the highest mode $w_a$, such that the ferromagnetic states correspond to the condensation of $n$ distinct modes present in
this interval, including the highest.

In addition, this will prove useful when analyzing the learning dynamics described in the next section.

\section{Learning dynamics of the RBM}\label{sec:rbmdyn}
\subsection{Learning dynamics in the thermodynamic limit}\label{sec:equations}
A mean field analysis of the learning dynamics has been proposed in~\cite{DeFiFu}, in the form of phenomenological equations obtained
after averaging over some parameters of the RBM, i.e. by choosing a well defined statistical ensemble of RBMs and using self-averaging properties
in the thermodynamic limit.
Here we rederive these equations, we add some details
and then explore their properties in the light of the preceding section.
First we project the gradient ascent equations~(\ref{eq:cd1}-\ref{eq:cd_theta})
onto the bases $\{u_\alpha(t)\in{\mathbb R}^{N_v}\}$ and $\{v_\alpha(t)\in{\mathbb R}^{N_h}\}$ defined by the SVD of $W$.
Discarding stochastic fluctuations usually inherent to the learning procedure and letting the learning rate $\gamma\to 0$,
the continuous version of~(\ref{eq:cd1}-\ref{eq:cd_theta}) can be recast as follows:
\begin{align}
\frac{1}{L}\Bigl(\frac{dW}{dt}\Bigr)_{\alpha\beta} &= \langle s_\alpha\sigma_\beta \rangle_{\rm Data}-\langle s_\alpha\sigma_\beta\rangle_{\rm RBM},\label{eq:csa1}\\[0.2cm]
\frac{1}{\sqrt{L}}\Bigl(\frac{d\eta}{dt}\Bigr)_{\alpha} &= \langle s_\alpha\rangle_{\rm RBM}-\langle s_\alpha\rangle_{\rm Data},\label{eq:csa2}\\[0.2cm]
\frac{1}{\sqrt{L}}\Bigl(\frac{d\theta}{dt}\Bigr)_{\alpha} &= \langle \sigma_\alpha\rangle_{\rm RBM}-\langle \sigma_\alpha\rangle_{\rm Data},\label{eq:csa3}
\end{align}
with $s_\alpha$ and $\sigma_\alpha$ given in~(\ref{eq:salpha}).
We also have
\begin{align*}
\left( \frac{dW}{dt} \right)_{\alpha \beta} &= \delta_{\alpha,\beta}\frac{dw_\alpha}{dt}+
(1-\delta_{\alpha,\beta})\Bigl(w_\beta(t)\Omega_{\beta\alpha}^v(t)+w_\alpha(t)\Omega_{\alpha\beta}^h\Bigr)\\[0.2cm]
\frac{1}{\sqrt{L}}\left( \frac{d\eta}{dt} \right)_{\alpha} &= \frac{d\eta_\alpha}{dt}-\sum_\beta \Omega_{\alpha\beta}^v\eta_\beta\\[0.2cm]
\frac{1}{\sqrt{L}}\left( \frac{d\theta}{dt} \right)_{\alpha} &= \frac{d\theta_\alpha}{dt}-\sum_\beta \Omega_{\alpha\beta}^h\theta_\beta
\end{align*}
where
\begin{align*}
\Omega_{\alpha \beta}^v (t) &= -\Omega_{\beta\alpha}^v \egaldef \frac{d\bm{u}^{\alpha,T}}{dt} \bm{u}^\beta  \\
\Omega_{\alpha \beta}^h (t)&= -\Omega_{\beta\alpha}^h \egaldef \frac{d\bm{v}^{\alpha,T}}{dt} \bm{v}^\beta
\end{align*}
By eliminating $\left( \frac{dw}{dt} \right)_{\alpha \beta}$, $\Bigl(\frac{d\eta}{dt}\Bigr)_{\alpha}$ and $\Bigl(\frac{d\theta}{dt}\Bigr)_{\alpha}$
we get the following set of dynamical equations:
\begin{align}
\frac{1}{L}\frac{dw_\alpha}{dt} &= \langle s_\alpha\sigma_\alpha \rangle_{\rm Data}-\langle s_\alpha\sigma_\alpha\rangle_{\rm RBM}\label{eq:dyn_w}\\[0.2cm]
\frac{d\eta_\alpha}{dt} &= \langle s_\alpha\rangle_{\rm RBM}-\langle s_\alpha\rangle_{\rm Data}+\sum_\beta \Omega_{\alpha\beta}^v\eta_\beta\label{eq:dyn_eta}\\[0.2cm]
\frac{d\theta_\alpha}{dt} &= \langle \sigma_\alpha\rangle_{\rm RBM}-\langle \sigma_\alpha\rangle_{\rm Data}+\sum_\beta \Omega_{\alpha\beta}^h\theta_\beta\label{eq:dyn_theta}
\end{align}
along with the infinitesimal rotation generators of the left and right singular vectors
\begin{align}
\Omega_{\alpha \beta}^v (t) &= -\frac{1}{w_\alpha+w_\beta}\left( \frac{dW}{dt} \right)_{\alpha \beta}^{\rm A} + \frac{1}{w_\alpha - w_\beta}\left( \frac{dW}{dt} \right)_{\alpha \beta}^{\rm S} \label{eq:dyn_rotv} \\[0.2cm]
\Omega_{\alpha \beta}^h (t)&=  \frac{1}{w_\alpha+w_\beta}\left( \frac{dW}{dt} \right)_{\alpha \beta}^{\rm A} + \frac{1}{w_\alpha - w_\beta}\left( \frac{dW}{dt} \right)_{\alpha \beta}^{\rm S} \label{eq:dyn_roth}
\end{align}
where
\[
\left( \frac{dW}{dt} \right)_{\alpha \beta}^{\rm A,S} \egaldef \frac{1}{2}
\Bigl(\langle s_\alpha \sigma_\beta \rangle_{\rm Data}\pm\langle s_\beta \sigma_\alpha \rangle_{\rm Data}
\mp\langle s_\beta \sigma_\alpha \rangle_{\rm RBM} - \langle s_\alpha \sigma_\beta \rangle_{\rm RBM}\Bigr).
\]
The dynamics of learning is now expressed in the reference frame defined by the singular vectors of $W$.
The skew-symmetric rotation generators $\Omega_{\alpha\beta}^{v,h}(t)$ of
the basis vectors (induced by the dynamics) tell us how data rotate relatively to this frame.
Given the initial conditions, these help us keeping track of the representation of data in this frame.
Note that these equations become singular when some degeneracy occurs in $W$ because then the SVD is not uniquely defined.
Except from the numerical point of view, where some regularizations might be needed,
this does not constitute an issue. In fact only rotations among non-degenerate modes are meaningful, while
the rest corresponds to gauge degrees of freedom.

At this point our set of dynamical equations~(\ref{eq:dyn_w}-\ref{eq:dyn_roth}) is written in a general form.
Our goal is to find the typical trajectory of the RBM within a certain statistical ensemble.
For this reason, we make the hypothesis that the learning dynamics is 
represented by a trajectory in the space $\{w_\alpha(t),\eta_\alpha(t),\theta_\alpha(t),\Omega_{\alpha\beta}^{v,h}(t)\}$, while the specific realization of $u_i^\alpha$,
$v_j^\alpha$  and $r_{ij}$ in~(\ref{eq:wsvd}) can be considered irrelevant and only the way they are distributed is important. We are then allowed to perform an 
average over $u_i^\alpha$, $v_j^\alpha$ and $r_{ij}$ with respect to some simple distributions, as long as 
this average is correlated with the data. By this we mean that the components $s_\alpha$ of any given sample are kept fixed while averaging. In the end, what really matters are the strength and the rotation of the SVD modes, respectively determined by $w_\alpha(t)$ and $\Omega_{\alpha\beta}^{v,h}(t)$.
As a simplification and also by lack of understanding of what intrinsically drives their evolution, 
the distributions of $u_i^\alpha$ and $v_j^\alpha$ will be considered stationary in the sequel. Concerning $r_{ij}$,
we allow its variance $\sigma^2/L$ to vary with time in order to give a minimal description of how the MP bulk evolves during the 
learning. The detailed dynamics of $\sigma$ will be derived later in Section~\ref{sec:non-linear}.
\begin{figure}[ht]
	\centering
	\includegraphics[scale=0.45]{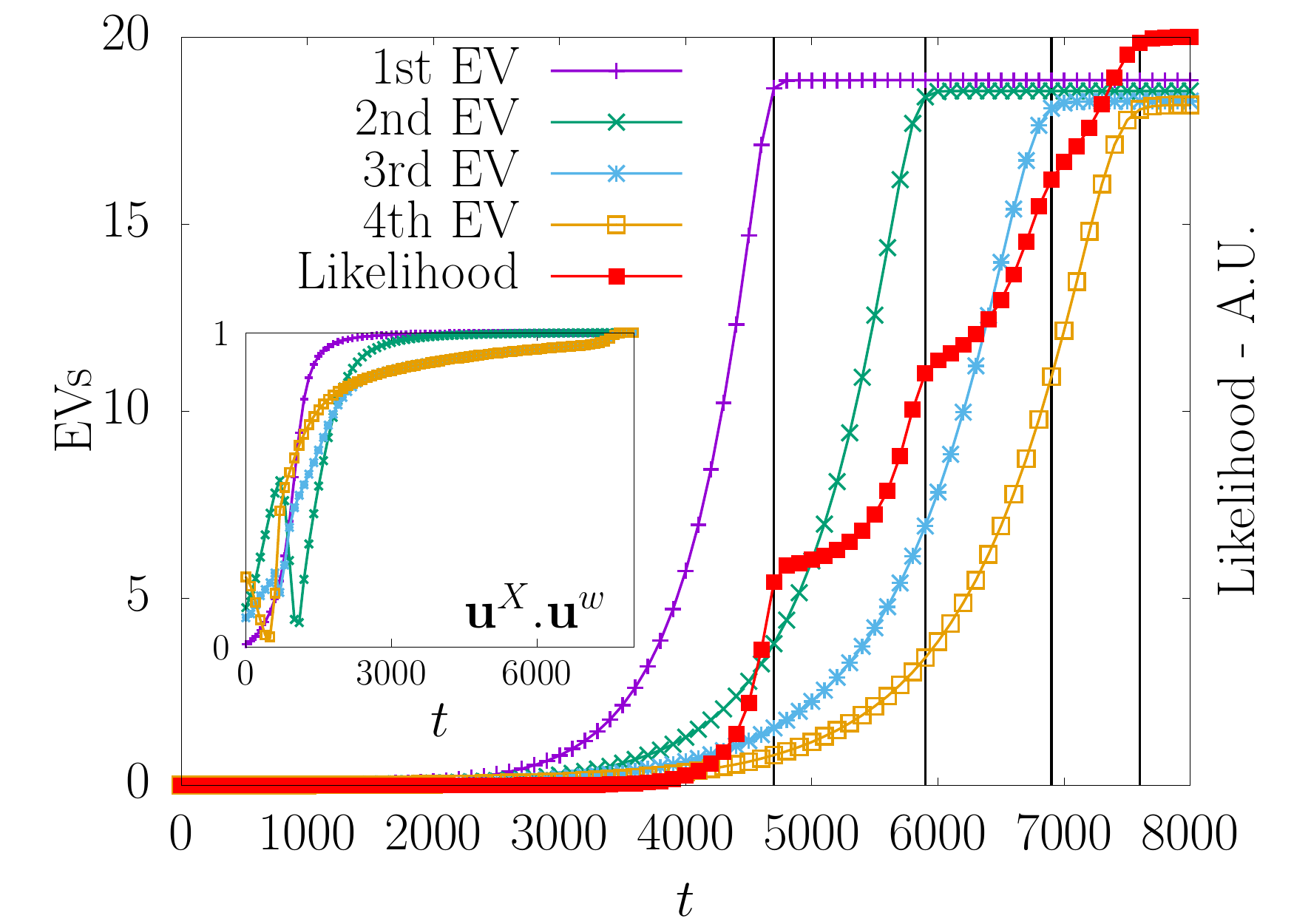}
	\caption{Time evolution of the eigenvalues and of the likelihood in the linear model. We observe very clearly how the different modes emerge from the bulk and how the likelihood increases with each learned eigenvalue. In the inset, the scalar product of the vectors $\bm{u}$ obtained from the SVD of the data and from the weights $\bm{w}$. The $\bm{u}$s of $\bm{w}$ are aligned with the SVD of the data at the end of the learning.}
	\label{fig_linRBM_dyn}
\end{figure}
Using the same notation of Section~\ref{sec:ferro_phase} and in particular using the rescaling  $v\sim \sqrt{N_h} v_i^\alpha$,
the empirical terms take the form:
\begin{align}
\langle \sigma_\alpha\rangle_{\rm Data} &= \langle (s_\alpha w_\alpha-\theta_\alpha)\bigl(1-q_\alpha[\bs]\bigr)\rangle_{\rm Data}\label{eq:empterm0}\\[0.2cm]
\langle s_\alpha\sigma_\beta\rangle_{\rm Data} &= \langle s_\alpha(s_\beta w_\beta-\theta_\beta)\bigl(1-q_\beta[\bs]\bigr)\rangle_{\rm Data}\label{eq:empterm}
\end{align}
where
\[
q_\alpha[\bs] \egaldef \int dx\frac{e^{-\frac{x^2}{2}}}{\sqrt{2\pi}} d\bv p_\alpha(\bv)\tanh^2
\Bigl(\kappa^{-\frac{1}{4}}\bigl(\sigma x+\sum_\gamma (w_\gamma s_\gamma-\theta_\gamma)v^\gamma\bigr)\Bigr),
\]
Note that the last equation actually depends on the activation function (hyperbolic tangent in this case), and the term $\sigma x$ corresponds to 
$\sum_k r_{kj}s_k$ and is obtained by central limit theorem from the independence of the $r_{kj}$. 
$q_\alpha[\bs]$ is the empirical counterpart of 
the EA parameters $q$ and $q_\alpha$ already encountered in Section~\ref{sec:ferro_phase}, and for simple i.i.d. distributions
like Gaussian or Laplace it can be estimated easily. The main point here is that the
empirical terms~(\ref{eq:empterm0},\ref{eq:empterm})
define operators whose decomposition over the SVD modes of $W$ functionally depends only on $w_\alpha,\theta_\alpha$
and on the projection of the data over the SVD modes of $W$.
These terms are driving the dynamics in a precise way. The adaptation of the RBM to this driving force is given by the ${\langle \dots \rangle_{\rm RBM}}$ terms
in~(\ref{eq:dyn_w},\ref{eq:dyn_eta},\ref{eq:dyn_theta}), which can be
estimated in the thermodynamic limit (see Section~\ref{sec:non-linear})
as a function of $w_\alpha$, $\theta_\alpha$ and $\eta_\alpha$ alone, by means of the order parameters
$(m_\alpha,\bar m_\alpha)$ given in Section~\ref{sec:RS} and once the mean-field equations~(\ref{eq:mf1},\ref{eq:mf2}) have been solved.
Of course, all of this is based on the hypothesis that the RBM stays in the RS domain during learning. Experimental evidence supports this hypothesis (see Section~\ref{sec:validation}).

\subsection{Linear instabilities}\label{sec:linear}
At the beginning of the learning, the elements of the weight matrix $W$ are usually small; therefore, we can analyze the linear behavior of the RBM in order to understand what happens.
In particular, we will see that the dynamics of a non-linear RBM at the beginning of the learning can be understood by looking at the stability analysis of the learning process. The purpose of this analysis is to identify which ``deformation modes'' of the weight matrix are the most unstable, and how they are related to the input 
data. Additionally, a good feature of the linear case 
is that no averaging is needed, the dynamics being actually independent on the particular realization of the components $u_i^\alpha$ and $v_j^\beta$. Also, always relative to the linear case, no
distinction has to be made between dominant modes and other modes to be treated as the noise component of equation~(\ref{eq:wsvd}), we can simply put all of the modes on the same footing.

Let us analyze the linear regime for an RBM with binary units. 
The derivation is done by rescaling all the weights and fields by a common ``inverse temperature'' $\beta$ and letting this go to zero 
in equation~(\ref{eq:cd1}). In principle, the stability analysis would lead to assume both the weights and the magnetizations to be small. 
However, we can assume only the magnetizations to be small and consider a slightly more general case with no approximations.
Such a case is analogous to a linear RBM whose magnetizations undergo Gaussian fluctuations, and it is derived by keeping up to quadratic terms of the magnetizations in the mean field free energy:
\begin{align*}
F_{MF}(\mu,\nu) &\simeq \frac{1}{2}\sum_{i=1}^N (1+\mu_i)\log(1+\mu_i)+(1-\mu_i)\log(1-\mu_i)\\[0.2cm]
&+\frac{1}{2}\sum_{j=1}^M (1+\nu_j)\log(1+\nu_j)+(1-\nu_j)\log(1-\nu_j)\\[0.2cm]
&-\sum_{i,j}\bigl(W_{ij}\mu_i\nu_j-\frac{1}{2}W_{ij}^2(\mu_i^2+\nu_j^2)\bigr)
+\sum_{i=1}^N\eta_i\mu_i+\sum_{j=1}^M\theta_j\nu_j \\[0.2cm]
&= \frac{1}{2\sigma_v^2}\sum_{i=1}^N \mu_i^2+\frac{1}{2\sigma_h^2}\sum_{j=1}^M \nu_i^2-\sum_{ij}W_{ij}\mu_i\nu_j
+\sum_{i=1}^N\eta_i\mu_i+\sum_{j=1}^M\theta_j\nu_j.
\end{align*}
where the variances $(\sigma_v^2,\sigma_h^2)$ of respectively visible and hidden variables read ($N_h<N_v$):
\begin{align}
\sigma_v^{-2} &= 1 + \sum_j W_{ij}^2 \simeq 1+\sum_\alpha w_\alpha^2\label{eq:sigma_v}\\[0.2cm]
\sigma_h^{-2} &= 1 + \sum_i W_{ij}^2 = 1+\sum_\alpha w_\alpha^2.\label{eq:sigma_h}
\end{align}
We omitted the quadratic term in $W_{ij}$ coming from the TAP contribution to the free energy, which is optional for our stability analysis.
In absence of this term the modes evolve strictly independently, while taking it into account leads to 
a correction to individual variances which couples the modes.

Magnetizations $(\mu,\nu)$ of visible and hidden variables have now Gaussian fluctuations with covariance matrix
\begin{equation*}
C(\mu_v,\mu_h) \egaldef
\left[
\begin{matrix}
\sigma_v^{-2} & - W \\[0.2cm]
-W^T & \sigma_h^{-2}
\end{matrix}
\right]^{-1}
\end{equation*}
We can discard the biases of the data and the related fields ($\theta_\alpha,\eta_\alpha)$ with a proper centering of the
variables, and we consider equation~(\ref{eq:dyn_w}) directly involving the
covariance matrix of the data expressed in the frame defined by the SVD modes of $W$
\begin{equation*}
\langle s_\alpha\sigma_\beta\rangle_{\rm Data} = \sigma_h^2 w_\beta \langle s_\alpha s_\beta\rangle_{\rm Data}.
\end{equation*}
From $C(\mu_v,\mu_h)$ we get the other terms yielding the following equations:
\begin{align*}
  \frac{dw_\alpha}{dt} &= w_\alpha\sigma_h^2\Bigl(\langle s_\alpha^2\rangle_{\rm Data} - \frac{\sigma_v^2}{1-\sigma_v^2\sigma_h^2w_\alpha^2}\Bigr) \\
  \Omega_{\alpha\beta}^{v,h} &=
(1-\delta_{\alpha\beta})\sigma_h^2\Bigl(\frac{w_\beta-w_\alpha}{w_\alpha+w_\beta}\mp\frac{w_\beta+w_\alpha}{w_\alpha-w_\beta}\Bigr)\langle s_\alpha s_\beta\rangle_{\rm Data}
\end{align*}
Note that these equations are exact for a linear RBM, since they can be derived without any reference to the coordinates of $u_\alpha$ and $v_\alpha$
over which we average in the non-linear regime. These equations tell us that the learning dynamics drives the rotation of the vectors
$\bm{u}^\alpha$ (and $\bm{v}^\alpha$) until they are aligned to the principal components of the data,
i.e. until $\langle s_\alpha s_\beta\rangle_{\rm Data}$ becomes diagonal.
Calling $\hat w_\alpha^2$ the empirical variance of the data,
the system reaches the following equilibrium values:
\begin{equation*}
w_\alpha^2 =
\begin{cases}
\DD \frac{\hat w_\alpha^2-\sigma_v^2}{\sigma_v^2\sigma_h^2 \hat w_\alpha^2}\qquad\ \ \text{if} \qquad \hat w_\alpha^2 > \sigma_v^2,\\
\DD 0\qquad\qquad\qquad \text{if} \qquad \hat w_\alpha^2 \le \sigma_v^2.
\end{cases}
\end{equation*}
assuming $(\sigma_v,\sigma_h)$ fixed.
From this we see that the RBM selects the strongest SVD modes of the data.
The linear instabilities correspond to directions along which the variance of the data is above the threshold $\sigma_v^2$, and
they determine the development of the unstable deformation modes of the weight matrix; during the learning process, these modes will eventually interact following the usual
mechanism of non-linear pattern formation encountered for instance in reaction-diffusion processes~\cite{Cross-Hohenberg}. Other possible deformations are damped to zero. The linear RBM will therefore learn all the principal
components that passed the threshold (up to $N_h$).
Note that this selection mechanism is already known to occur for linear auto-encoders~\cite{Bou-Kamp}
or other similar linear Boltzmann machines~\cite{TiBi}.
On Fig.~\ref{fig_linRBM_dyn} we can see the eigenvalues being learned one by one in a linear RBM. 
 
If we take into account the expressions~(\ref{eq:sigma_v},\ref{eq:sigma_h}) for $(\sigma_v,\sigma_h)$, we see that the system cannot reach 
a stable solution except for the case in which all the modes are below the threshold at the beginning. Otherwise the modes that are excited first will eventually grow
like $\sqrt{t}$ for a large time, and the excitation threshold will tend to zero for all modes.

In any case, by the definition of a multivariate Gaussian, this simple non-linear analysis 
describes a unimodal distribution. In order to properly understand the dynamics and the steady-state regime
of a non-linear RBM, a well suited mean-field theory is required.
\begin{figure}
	\centering
	\includegraphics[scale=0.7]{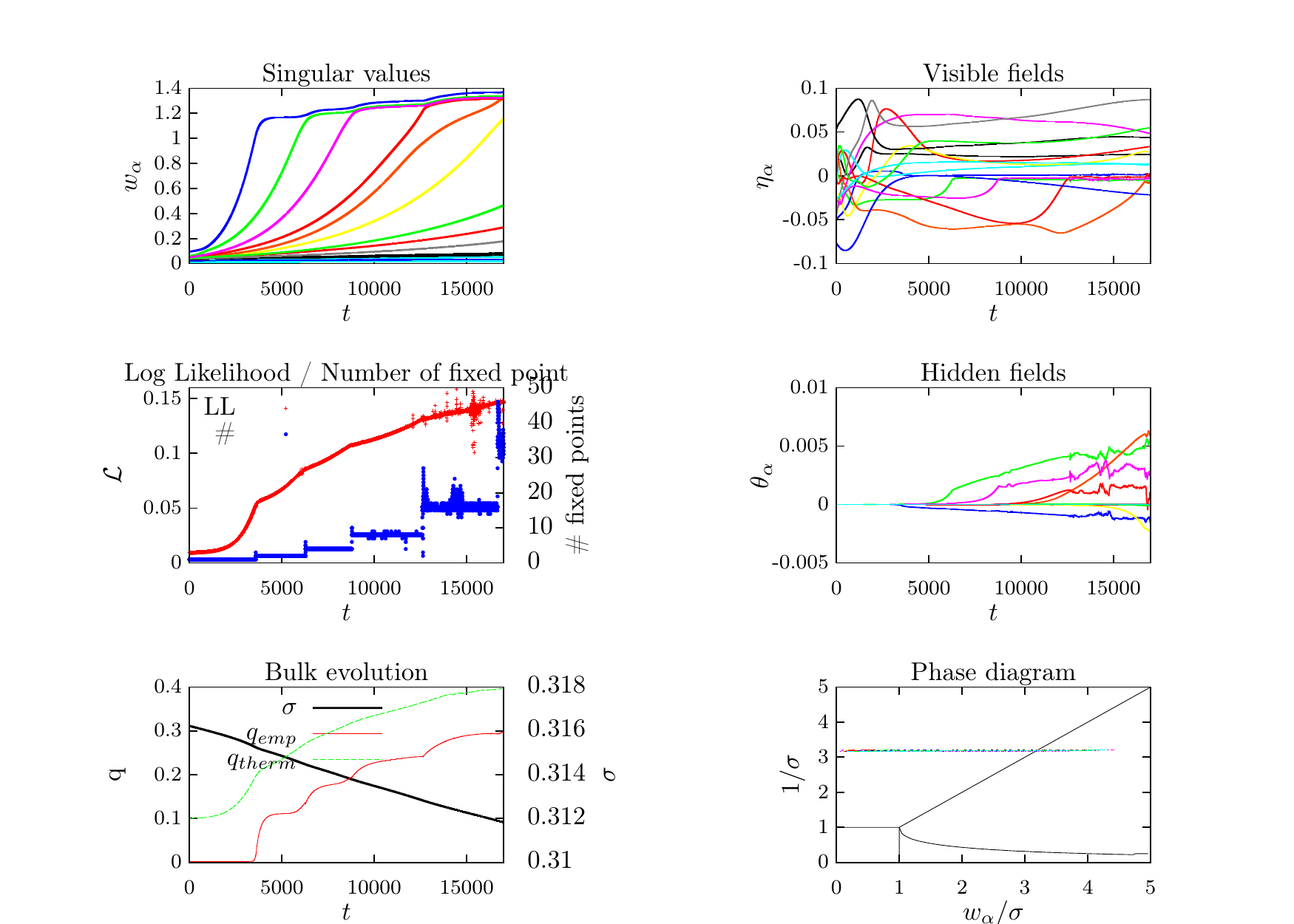}
\caption{Predicted mean evolution of an RBM of size $(N_v,N_h)=(1000,500)$ learned on a synthetic dataset of $10^4$ samples of size $N_v= 1000$
obtained from a multimodal distribution with $20$ clusters randomly defined on a submanifold of dimension $d=15$. The dynamics follows the projected
magnetizations in this reduced space with help of $15$ modes. We observe a kind of pressure on top singular values from lower ones.\label{fig:learning}}
\end{figure}

\subsection{Non-linear regime}\label{sec:non-linear}
In the linear regime, some specific modes are selected and at some point they start to interact in a non-trivial manner. As seen explicitly in~(\ref{eq:empterm}), the empirical
terms in~(\ref{eq:cd1}-\ref{eq:cd_theta}) involve higher order statistics of the data
 and then the Gaussian estimation with $\sigma_v^2=\sigma_h^2=1$
of the RBM response terms $\langle s_\alpha\rangle_{\rm RBM}$ and $\langle s_\alpha\sigma_\beta\rangle_{\rm RBM}$ is no longer valid when the interactions kick in.
Schematically, the linear regime is valid as long as the RBM is found in the paramagnetic phase. But as soon as one mode passes the linear threshold, 
the system enters the ferromagnetic phase. Then the 
proper estimation of the response terms follows from the thermodynamic analysis performed in Section~\ref{sec:thermo}, and depends on the  
assumptions made on the statistical properties of the components of the singular vectors of the weight matrix. 
In the case of Gaussian i.i.d. components, given the analysis proposed in Section~\ref{sec:ferro_phase},
we know that the mode with the highest singular value completely dominates the ferromagnetic phase: we expect one single ferromagnetic state
characterized by magnetizations aligned to this mode only, while magnetizations correlated to other modes vanish. To be precise, this is the correct
picture without fields ($\eta=\theta=0$) but we don't expect this picture to drastically change in the case of non-vanishing fields. In fact,
solving the mean-field equations in presence of the fields show the appearance of meta-stable states correlated with single dominated modes;
however, the free energy difference with respect to the ground state, i.e. the state correlated with the mode with the highest singular value,
is of order $O\bigl(L(w_\alpha-w_{max})\bigr)$, which means that the contribution of those meta-stable states become rapidly negligible with large system size.

To draw a realistic picture of the learning process we now consider Laplace i.i.d. components for the SVD modes that,
as seen in Section~\ref{sec:ferro_phase}, allow the ferromagnetic phase to be of compositional type. The reason for this is that the Laplace 
distribution leads to less interference among modes than the Gaussian distribution, so that the modes will weakly interact in the mean-field equations.
Solving equations (\ref{eq:mf1_alpha},\ref{eq:mf2_alpha},\ref{eq:q_exp},\ref{eq:qalpha_exp}) in absence of fields yields the following 
picture: one fixed point solution will typically have non-vanishing magnetizations $\{m_\alpha,\bar m_\alpha\}$ for all 
$\alpha$ such that $w_\alpha\in[w_{max}-\Delta w, w_{max}]$, where $\Delta w$ is approximately the gap $\Delta w(q,\bar q)$
defined in (\ref{eq:gap}). This solution is a degenerate ground state, all other solutions being obtained by independently reversing  
the signs of the condensed magnetizations $(m_\alpha,\bar m_\alpha)$. Hence for $K$ condensed modes we get a degeneracy
of $2^K$. When the fields are included, all these fixed points are displaced in the direction of the fields, and some of them may disappear. In the end we are left with a potentially large amount of nearly degenerate states able to cover the empirical distribution of the data, at least in some simple cases.
\begin{figure}
	\centering
	\includegraphics[scale=0.6]{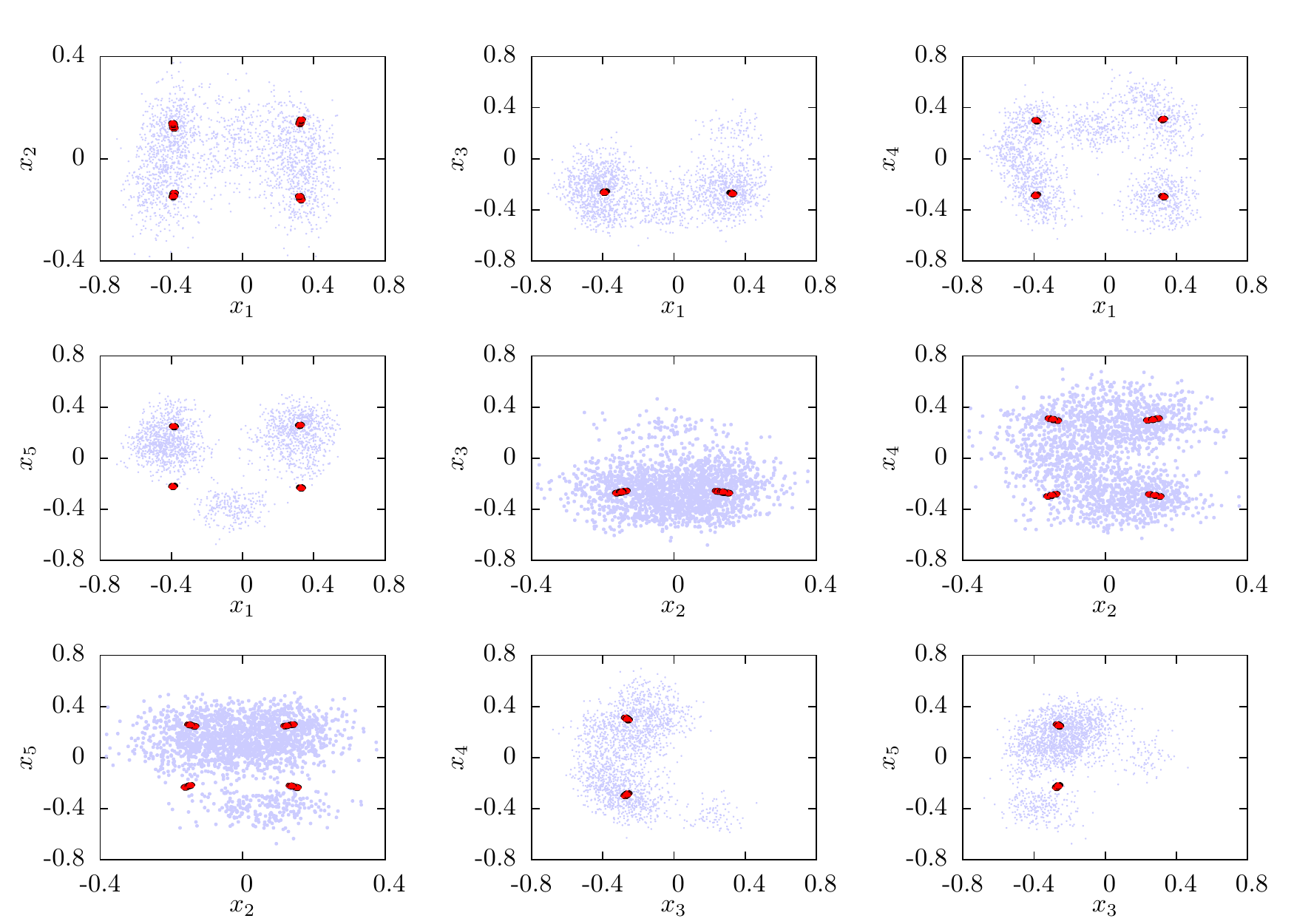}
\caption{Scatter plots of the mean-field magnetizations (in red) and the samples (in blue) in various plan projections defined by pairs of left eigenvectors of
$W$. This case corresponds to an RBM of size $(N_v,N_h)=(100,50)$ learned on a synthetic dataset of $10^4$ samples of size $N_v= 100$
obtained from a multimodal distribution with $11$ clusters randomly defined on a submanifold of dimension $d=5$. The scatter plot is obtained at a point where
$5$ modes have already condensed and $16$ saddle point solutions have been found.\label{fig:scatter}}
\end{figure}

Coming back to the learning dynamics the terms corresponding to the response of the RBM in~(\ref{eq:cd1},\ref{eq:cd_theta}) are estimated in the thermodynamic limit by means of
the previously defined order parameters:
\begin{align*}
\langle s_\alpha \rangle_{\rm RBM} &= \frac{1}{Z_{\rm Therm}}\sum_\omega e^{-Lf(m^\omega,\bar m^\omega,q^\omega,\bar q^{\omega})}\bar m_\alpha^\omega  
\egaldef \langle \bar m_\alpha \rangle_{\rm Therm},\\[0.2cm]
\langle s_\alpha\sigma_\beta\rangle_{\rm RBM} &= 
\frac{1}{Z_{\rm Therm}}\sum_\omega e^{-Lf(m^\omega,\bar m^\omega,q^\omega,\bar q^{\omega})}\bar m_\alpha^\omega m_\beta^\omega \egaldef \langle \bar m_\alpha m_\beta\rangle_{\rm Therm}.  
\end{align*}
Here $\langle \dots \rangle_{\rm Therm}$ denotes the thermodynamical average
and the partition function is expressed, in the thermodynamic limit, as   
\[
Z_{\rm Therm} \egaldef \sum_\omega e^{-Lf(m^\omega,\bar m^\omega,q^\omega,\bar q^{\omega})}
\]
The index $\omega$ runs over all the stable fixed point solutions of~(\ref{eq:mf1},\ref{eq:mf2}) weighted accordingly to the free energy given by~(\ref{eq:freen}).
These are the dominant contributions as long as free energy differences are $O(1)$, 
and the internal fluctuations given by each fixed point are comparatively of order $O(1/L)$. 
In addition, the dynamics of the bulk can be characterized by empirically defining $\sigma^2$:

\[
\sigma^2 = \frac{1}{L}\sum_{ij}r_{ij}^2, 
\]

whose evolution is:

\begin{align*}
\frac{d\sigma^2}{dt} &= \frac{1}{L}\sum_{ij}r_{ij}\frac{dW_{ij}}{dt},\\[0.2cm]
&=\frac{1}{L}\sum_{ij}r_{ij}\left[\langle s_i\tanh\Bigl(\sum_k r_{kj}s_k +\kappa^{-\frac{1}{4}}\sum_\alpha (w_\alpha s_\alpha-\theta_\alpha)v_j^\alpha\sqrt{L}\Bigr)\rangle_{\rm Data}
-\langle s_i\sigma_j\rangle_{\rm RBM}\right]
\end{align*}
given the independence of $r_{i*}$ (resp. $r_{*j}$) and $u_i^\alpha$ (resp. $v_i^\alpha$).

Exploiting the self-averaging properties of both the empirical and the response terms 
with respect to $r_{ij}$, $u_i^\alpha$ and $v_j^\alpha$ yields
\begin{align*}
\frac{1}{L^2}\sum_{ij} r_{ij}\langle s_i\sigma_j\rangle_{\rm Data} &= \frac{\sigma^2}{L}\bigl(1-\langle q[\bs]\rangle_{\rm Data}\bigr)\\[0.2cm] 
\frac{1}{L^2}\sum_{ij} r_{ij}\langle s_i\sigma_j\rangle_{\rm RBM}  &= \frac{\sigma^2}{L}\bigl(1- \langle q\rangle_{\rm Therm}\bigr),
\end{align*}
with 
\[
q[\bs] \egaldef \int dx\frac{e^{-\frac{x^2}{2}}}{\sqrt{2\pi}} d\bv p(\bv)\tanh^2
\Bigl(\kappa^{-\frac{1}{4}}\bigl(\sigma x+\sum_\gamma (w_\gamma s_\gamma-\theta_\gamma)v^\gamma\bigr)\Bigr).
\]
Summarizing, our equations take the suggestive form
\begin{align}
\frac{1}{L}\frac{dw_\alpha}{dt} &= \langle s_\alpha(w_\alpha s_\alpha-\theta_\alpha)(1-q_\alpha[\bs])\rangle_{\rm Data} 
-\langle \bar m_\alpha(w_\alpha\bar m_\alpha-\theta_\alpha)(1-q_\alpha)\rangle_{\rm Therm},\label{eq:dyn_therm_w}\\[0.2cm]
\frac{d\eta_\alpha}{dt} &= \langle\bar m_\alpha\rangle_{\rm Therm}-\langle s_\alpha\rangle_{\rm Data}+\sum_\beta 
\Omega_{\alpha\beta}^v\eta_\beta,\label{eq:dyn_therm_eta}\\[0.2cm]
\frac{d\theta_\alpha}{dt} &= \langle(w_\alpha\bar m_\alpha-\theta_\alpha)(1-q_\alpha)\rangle_{\rm Therm}
-\langle (w_\alpha s_\alpha-\theta_\alpha)(1-q_\alpha[\bs])\rangle_{\rm Data}+\sum_\beta \Omega_{\alpha\beta}^h\theta_\beta,\label{eq:dyn_therm_theta}\\[0.2cm]
\frac{d\sigma^2}{dt} &= \sigma^2\Bigl(\langle q\rangle_{\rm Therm} -\langle q[\bs]\rangle_{\rm Data}\Bigr),\label{eq:dyn_therm_sigma}
\end{align}
with $\Omega^{v,h}$ taking the form of a difference between a data averaging $\langle \dots \rangle_{\rm Data}$ and 
a thermodynamical averaging $\langle \dots \rangle_{\rm Therm}$ involving only order parameters. Note here that the $w_\alpha$ variables, with respect to the other variables, evolve on a faster time scale.
This is our final and main result, which might possibly help improving current learning algorithms of RBMs.
From this, it is clear what the learning of an RBM is aimed at: the equations will converge once the dataset is clustered in such a way
that each cluster is represented by a solution of the mean-field equations with magnetizations  $\bar m_\alpha$ and EA parameters $q_\alpha$ 
corresponding respectively to their empirical counterparts $\langle s_\alpha \rangle$ and $\langle q_\alpha[\bs]\rangle$ representing cluster magnetization and variance. 
In particular, these clusters can somehow be regarded as the attractors in the context of feed-forward networks, defining a partition of the data. This can be seen by starting from random configurations and letting the system evolve using the TAP equations or a MCMC method. At the end the system will end up in one of those clusters (characterized by a fixed point of the mean-field equations). Note that this is the reason why the RBM needs to reach a ferromagnetic phase with many states to be able to match the empirical term in~(\ref{eq:cd1}) and reach convergence.

Additionally, the log likelihood~(\ref{eq:LL}) can be estimated in the thermodynamic limit
(after normalization by $L$).
\begin{align*}
{\mathcal L} &= \Big\langle\sqrt{\kappa}\EE_{x,v}
\Bigl[\log\cosh\Bigl(\kappa^{-\frac{1}{4}}\bigl(\sigma x+\sum_\alpha(w_\alpha s_\alpha-\theta_\alpha)v^\alpha\bigr)\Bigr)\Bigr]\Big\rangle_{\rm Data} \\[0.2cm]
&-\big\langle\sum_\alpha \eta_\alpha s_\alpha\big\rangle_{\rm Data}-\frac{1}{L}\log\big(Z_{\rm Therm}\bigr),
\end{align*}

As an example, for a multimodal data distribution with a finite number of clusters embedded in a high dimensional configuration space, 
the SVD modes of $W$ that will develop are the one pointing to the directions of the magnetizations 
defined by these clusters (which will be almost surely orthogonal, given the high dimensionality of the embedding space).
In this simple case the RBM will evolve, as in the linear case, to a state in which the empirical term  becomes diagonal, while the singular values 
will adjust to match the proper magnetization in each fixed point.

We have integrated equations~(\ref{eq:dyn_therm_w},\ref{eq:dyn_therm_eta},\ref{eq:dyn_therm_theta},\ref{eq:dyn_therm_sigma},\ref{eq:dyn_rotv},\ref{eq:dyn_roth}) 
in simple cases by using the Laplace averaging of the components of the SVD modes and using for the EA parameters the expressions given in (\ref{eq:q_exp},\ref{eq:qalpha_exp}).
Basically, the hidden distribution to be modeled is 
defined by
\begin{equation}
P(s) = \sum_{c=1}^C p_c\prod_{i=1}^N \frac{e^{h_i^c s_i}}{2\cosh(h_i^c)},
\label{eq:clusters}
\end{equation}
i.e. a multimodal distribution composed of $C$ clusters of independent variables, where the magnetization of each variable $i$
in cluster $c$ is given by $m_i^c = \tanh(h_i^c)$. Each cluster is weighted by some probability $p_c$.
In addition we assume these magnetization vectors $m^c$ to be embedded in a low dimensional space
of dimension $d<<N$. $d$ defines the rank of $W$.
The initial conditions for $W$ are such that the left singular vectors $\{u_\alpha,\alpha=1,\ldots d\}$ span this low
dimensional space. An example of the typical dynamics obtained in the case at hand is shown in Figure~\ref{fig:learning}. In contrast to the
linear problem where singular values evolve independently, here we distinctively witness the interaction between singular values: 
a kind of pressure is exerted by lower modes on higher ones resulting in successive
bumps in the dynamics of the top modes. The number of states is roughly multiplied by two each time a mode condenses and get close enough
to the top modes. Concerning the dynamics of the fields, we don't really observe convergence towards stable directions. Some (possibly numerical)
instability is observed when many modes condense, with both the fields and the number of fixed point solutions becoming very noisy.
It is also interesting to see how the magnetizations related to the states are distributed with respect to the dataset.
On Figure~\ref{fig:scatter} we see that the fixed points tend (as expected) to settle within dense regions of sample points.
However, our coarse description shows some limitations for more complex situations, the number of adjustable parameters being too limited to be able 
to match arbitrary distributions of clusters. 
It is then appropriate to think about this behaviour in a mean sense; at least,
it is able to reproduce a realistic 
learning dynamics of the singular values of the weight matrix.

\begin{figure}
  \includegraphics[width=\linewidth]{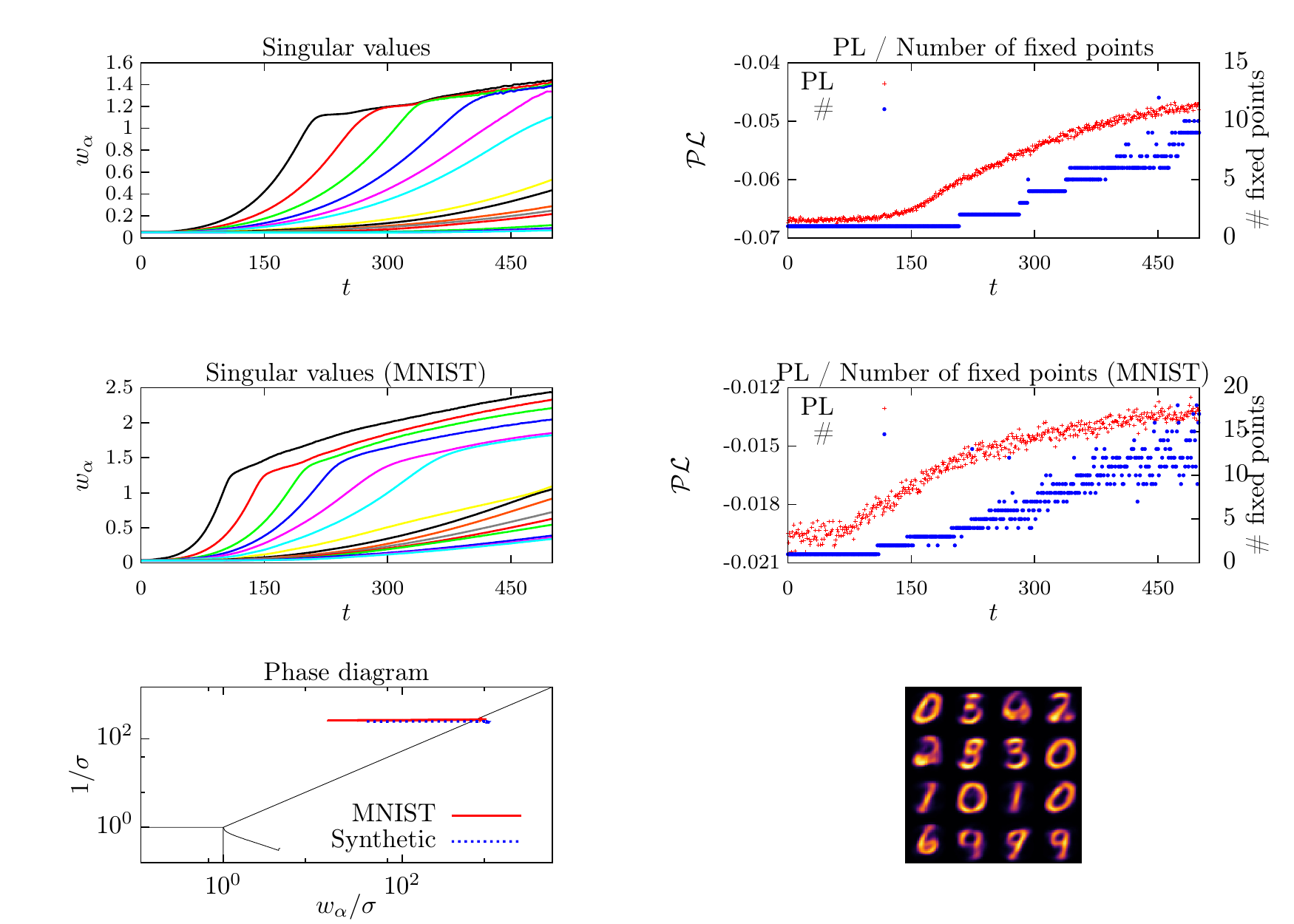}
	\caption{Experimental evolution of an RBM during training for a synthetic dataset (top plots, to compare to Fig.~\ref{fig:learning}) and for MNIST (central plots). The bottom left plot shows the learning trajectories in the phase diagram, while the bottom right image shows some examples of fixed point solutions for MNIST (we note the presence of some spurious fixed points).}
  \label{fig:plot_expt}
\end{figure}

\begin{figure}
  \centering
  \includegraphics[width=.5\linewidth]{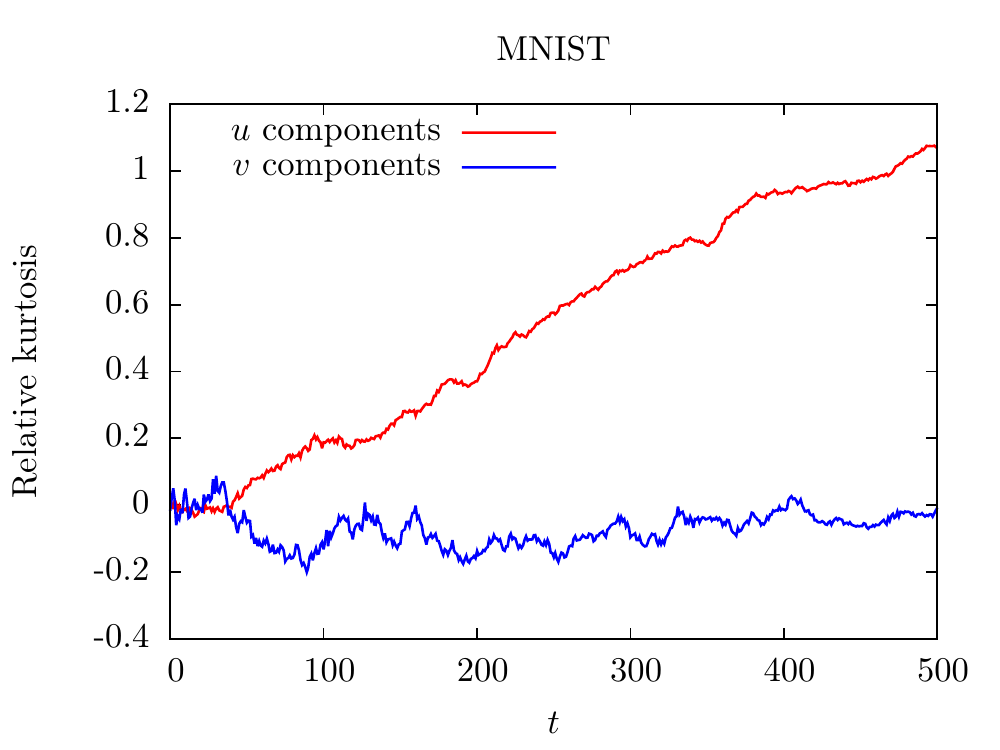}
	\caption{Relative kurtosis of the components of the modes after training on MNIST.}
  \label{fig:kurtosis}
\end{figure}

\section{Numerical Experiments} \label{sec:validation}
Given the comprehensive theoretical analysis of the RBM model given in the previous sections, we are now able to provide a meaningful description of the learning dynamics for a RBM trained with k-steps contrastive divergence (CDk) ~\cite{Hinton_CD}. The observations presented in this section will serve as a validation for the theoretical analysis. First, to provide a more direct comparison to section \ref{sec:non-linear}, we will look at the learning dynamics of an RBM trained on a set of simple synthetic data. Subsequently, we will test the model against real world data by training on the MNIST dataset.


\subsection{Synthetic dataset}
As a simple case, we trained the RBM over the same dataset defined in fig.~\ref{fig:learning}, derived from the simple multimodal distribution in eq.~\ref{eq:clusters} (see Appendix ~\ref{app:data} for details). Thus we set \( N_v = 1000 \), \(N_h = 500 \) and we trained using \(10^4\) samples with an effective dimension \(d = 15\) organized in \(20\) separate clusters. The weights are initialized from a Gaussian distribution with standard deviation \( \sigma = 10^{-3} \), while the hidden bias is initialized to \(0\) and the visible bias is initialized with the empirical mean of the data

\[
\eta_i = \frac{1}{2} \log \left( \frac{p_i}{1 - p_i} \right)
\]
where \(p_i\) is the empirical probability of activation for the \(i_{th}\) hidden node.

Finally, the training set is divided into batches of size \(20\), 5 Gibbs sampling steps are used (CD5) and the 
learning rate \(\gamma\) is kept low in order to reduce noise, \(\gamma = 5 \times 10^{-8}\).
The results of the analysis are shown in fig.~\ref{fig:plot_expt}. We see that the dynamics of the singular values obtained by direct integration of the mean-field equations (Fig.~\ref{fig:learning}) are very well reproduced, the only difference being a slightly higher pressure on the strongest modes. The number of fixed point solutions also seems to follow the same trend but more noise is present, an indication of the fact that the RBM has a tendency to learn spurious fixed points during the training. The learning trajectory on the phase diagram is also of interest; we see that the RBM is initialized in the paramagnetic state as expected and the effect of the learning is to drive the model to the ferromagnetic phase. Once in the ferromagnetic phase, the trajectory slows down and the model is assessed near the critical line between paramagnetic and ferromagnetic states, where the estimate of the weights is most stable (according to~\cite{MM_criticality}). Finally, in Fig.~\ref{fig:scatter_expt} we see how the RBM is able to generate a proper clustering of the data over the spectral modes. In particular, the TAP fixed points of the trained model are well distributed and able to cover the full data distribution, improving over the typical behaviour for Laplace distributed weights that emerged with our theoretical analysis (Fig.~\ref{fig:scatter}).

\begin{figure}
  \includegraphics[width=\linewidth]{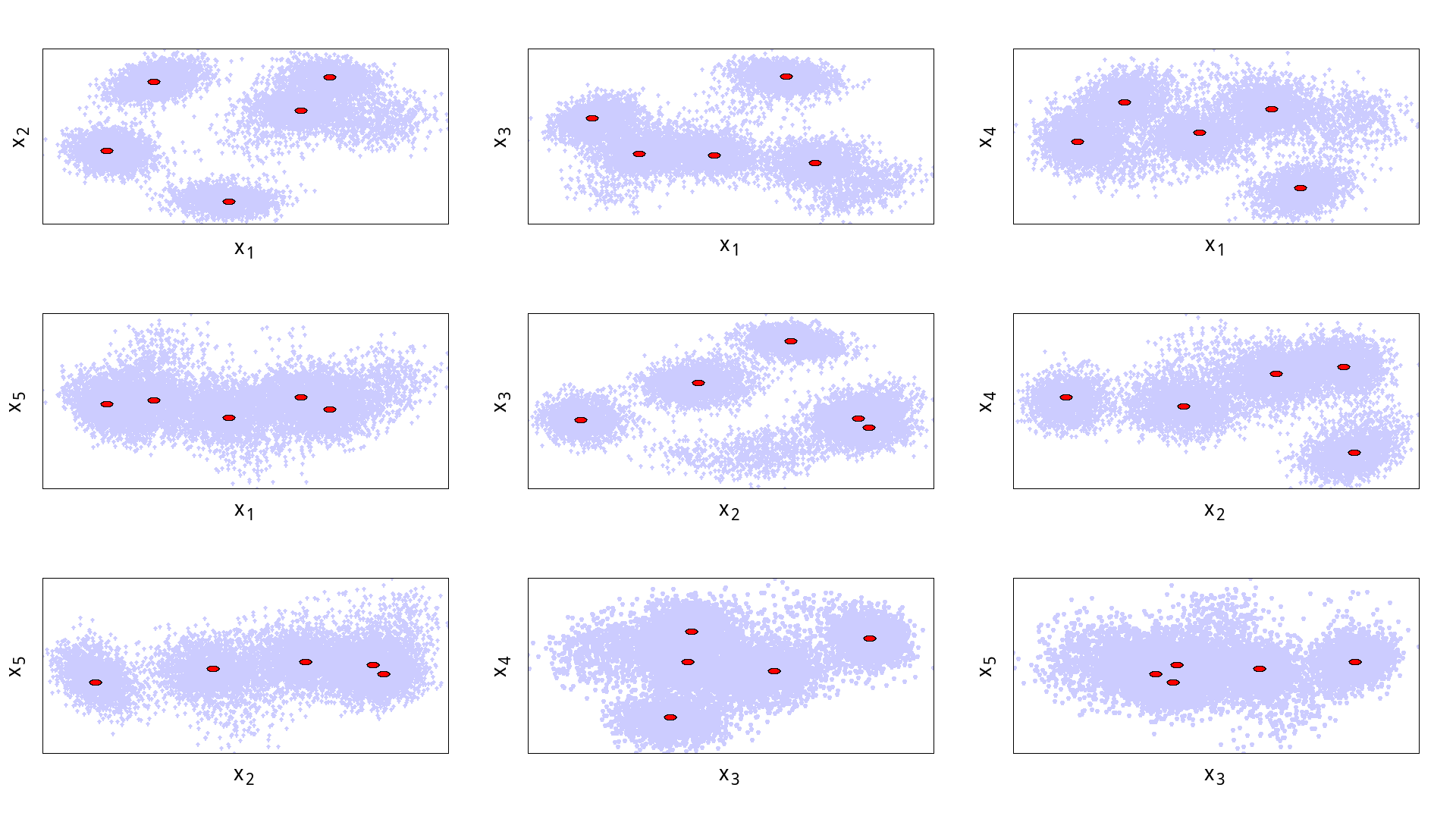}
	\caption{Scatter plots of samples (blue) and fixed points (red) in various plan projections defined by pairs of left eigenvectors of W. The dataset is the same as in Fig.~\ref{fig:scatter} and in this case 5 modes have condensed and 7 fixed point solutions have been found.}
  \label{fig:scatter_expt}
\end{figure}

\subsection{MNIST dataset}
The MNIST dataset is composed by 70000 handwritten digits (60000 for training, 10000 for testing) of size \(28\times28\) pixels. Being highly multimodal, we expect this dataset  to push the limits of our spectral analysis. For the training, the initialization of the model is the same one used for the synthetic data, 10000 training samples are used (taken at random from the dataset) and the values of the other hyperparameters are as follows: \( N_v = 784 \), \(N_h = 100\), batch size \( = 20 \), \(\gamma = 5 \times 10^{-7}\). With respect to the linear regime (described in section \ref{sec:linear}) we see in Fig. \ref{fig_plot_svdmode} how the RBM is able to learn the SVD of the dataset quite precisely at the beginning of the training, then the learning dynamics quickly enter the non-linear regime. Even in this highly multimodal scenario, our findings over simple synthetic data seem to be confirmed, as seen in Fig.~\ref{fig:plot_expt}. The high number of modes, however, determines an increase in the magnitude of the singular values of condensed modes and seems to destabilize a bit the learning, making the computation of fixed points less reliable. In fact, as a high number of modes are condensing, the model is not able to get rid of all the spurious fixed points. This problem can be mitigated by using an even smaller learning rate, at the cost of slowing down the training. Probably, using a variable learning rate could be a more practical solution (descreasing the learning rate from time to time to let the model eliminate unneeded fixed points).
Concerning the (relative) kurtosis of the mode components distributions, we did not observe a very stable and systematic behavior. Either we see small fluctuations around zero, either some excursions occur and a finite value in the range $[0,3]$ is building up either for the $u$ or the $v$ components, coherently to the compositional phase interpretation given previously. The latter is the case for MNIST, as shown in Fig.~\ref{fig:kurtosis}. Additionally the transverse part of the fields, meaning orthogonal to the condensed modes, is usually not completely negligible, in contrary to what we assume in~(\ref{def:eta},\ref{def:theta}). This clearly constitutes a limitation of our analysis. These transverse components offer more flexibility for generating and selecting fixed points and interfere in some non-trivial way with the kurtosis property, which possibly explains why we don't get a systematic behavior.

\begin{figure}
  \begin{subfigure}{\linewidth}
    \centering
    \includegraphics[width=.08\linewidth]{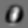}
    \includegraphics[width=.08\linewidth]{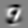}
    \includegraphics[width=.08\linewidth]{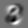}
    \includegraphics[width=.08\linewidth]{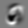}
    \includegraphics[width=.08\linewidth]{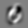}
    \includegraphics[width=.08\linewidth]{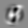}
    \includegraphics[width=.08\linewidth]{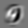}
    \includegraphics[width=.08\linewidth]{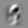}
    \includegraphics[width=.08\linewidth]{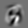}
    \includegraphics[width=.08\linewidth]{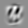}
    \caption{}
    \label{fig:modes_data}
  \end{subfigure}
  \begin{subfigure}{\linewidth}
  	\centering
    \includegraphics[width=.08\linewidth]{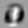}
    \includegraphics[width=.08\linewidth]{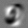}
    \includegraphics[width=.08\linewidth]{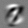}
    \includegraphics[width=.08\linewidth]{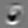}
    \includegraphics[width=.08\linewidth]{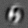}
    \includegraphics[width=.08\linewidth]{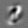}
    \includegraphics[width=.08\linewidth]{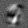}
    \includegraphics[width=.08\linewidth]{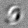}
    \includegraphics[width=.08\linewidth]{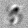}
    \includegraphics[width=.08\linewidth]{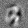}
    \caption{}
    \label{fig:modes_tr}
  \end{subfigure}\par
  \begin{subfigure}{\linewidth}
    \centering
    \includegraphics[width=.08\linewidth]{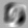}
    \includegraphics[width=.08\linewidth]{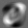}
    \includegraphics[width=.08\linewidth]{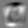}
    \includegraphics[width=.08\linewidth]{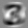}
    \includegraphics[width=.08\linewidth]{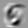}
    \includegraphics[width=.08\linewidth]{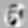}
    \includegraphics[width=.08\linewidth]{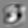}
    \includegraphics[width=.08\linewidth]{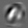}
    \includegraphics[width=.08\linewidth]{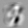}
    \includegraphics[width=.08\linewidth]{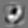}
    \caption{}
    \label{fig:modes_tr_10}
  \end{subfigure}
   \caption{\textbf{(a)} Principal components extracted from the training set (starting from the second, as the first one is encoded into the visible bias). \textbf{(b)} The first 10 modes of a RBM trained for 1 epoch (with \(\gamma \simeq 0.1\)). \textbf{(c)} Same as (b) but after a 10 epochs training.}
  \label{fig_plot_svdmode}
\end{figure}

\section{Discussion}
Before drawing some perspectives, let us summarize the main outcomes of the present work:
\begin{itemize}
\item {\bf(i) thermodynamic properties of realistic RBMs:} 
our analysis focused on a non-i.i.d. ensemble of weight matrices, whose derivation has been inspired by empirical observations obtained by training RBMs on real data.
\item {\bf(ii) RS equations and compositional phase}: we found a way of writing the RS equations for the RBM (in particular with equations~(\ref{eq:mf1_alpha},\ref{eq:mf2_alpha},\ref{eq:mf3_alpha},\ref{eq:mf4_alpha})) which leads to a simple characterization of the ferromagnetic phase where the RBM is assumed to operate.
Schematically, a negative relative kurtosis for the distribution of the singular vectors' components favors the proliferation of metastable states, while a positive one tends to 
favor a compositional phase. In particular, we were able to precisely address a concrete case presenting the compositional phase by considering a Laplace distribution
for the singular vectors' components.
\item {\bf(iii) a set of equations representing a typical learning dynamics} 
that defines a trajectory in $\{w_\alpha(t),\eta_\alpha(t),\theta_\alpha(t),\Omega_{\alpha\beta}^{v,h}(t),\sigma^2(t)\}$.
The spectrum of the dominant singular values, represented by $\{w_\alpha(t)\}$ and expressing the information content of the RBM, is playing the main role.
The bulk of dominated modes corresponding to noise sees its dynamics summarized by the evolution of $\sigma^2(t)$. Rotations of dominant singular
vectors during the learning process are given by $\Omega^{v,h}$ while the projections of the biases along the main modes are given by $\eta$ and $\theta$. 
These equations have been obtained by averaging over the components of left and right SVD vectors of the weight matrix, while keeping fixed the quantities considered to be relevant. This averaging actually corresponds to a standard self-averaging assumption in a RS phase.
\item {\bf(iv) a clustering interpretation of the training process} is obtained through 
equations~(\ref{eq:dyn_therm_w},\ref{eq:dyn_therm_eta},\ref{eq:dyn_therm_theta},\ref{eq:dyn_therm_sigma}) where it is explicitly shown the kind of matching that
the RBM is trying to perform between the order parameters obtained from the fixed point solutions and their empirical counterparts in the non-linear regime. 
A natural clustering of the data can actually be defined by assigning to each sample the fixed point obtained after initializing the fixed point equations with a
visible configuration corresponding to that same sample.
\end{itemize}
The main picture emerging from the present analysis is that of a set of clusters corresponding to the fixed points of the RBM, which try to uniformly cover the support of the dataset. A full understanding of the mechanism by which the RBM manages 
to properly cover the dataset is still lacking, even though the case of Laplace distributed singular vectors' components gives
some insights. By comparison, real RBMs have more flexibility than the simple ``mean Laplace RBM'' considered in Section~\ref{sec:ferro_phase}
and they can produce a good covering of the data manifold. We were not yet able to precisely pinpoint the main ingredients for that mechanism, 
even though we suspect the transverse biases (orthogonal to the modes) of the hidden units 
to be the missing ingredient in our analysis. 

From the theoretical point of view we would like to see how these results can be adapted to more complex models like DBM or generative models based on 
convolutional networks. In particular we would like to understand whether adding more layers can facilitate the covering of the dataset by fixed points.
From the practical point of view these results might help to orientate the choice of the hyper-parameters used for training an RBM
and to refine the criteria for assessing the quality of a learned RBM. For instance, the choice of the number of hidden variables is dictated 
by two considerations: the effective rank of $W$, i.e. the number of relevant modes to be considered, and the level of interaction between these modes. 
Using less hidden variables gives more compact RBMs and reduces the rank of $W$ to its needed value,  
but it also leads to modes with stronger interactions, which means less flexibility for generating a good covering of fixed points.

\appendix

\section{AT line}\label{app:AT}

The stability of the RS solution to the mean-field equations is studied along the lines of~\cite{AlTh} by looking at the Hessian
of the replicated version of the free energy and identifying eigenmodes from symmetry arguments.
Before taking the limit $p\to 0$ the free energy reads
\[
f[m,\bar m,Q,\bar Q] = \sum_{a,\alpha}w_\alpha m_\alpha^a\bar m_\alpha^a +\frac{\sigma^2}{2}\sum_{a\ne b} Q_{ab}\bar Q_{ab}
-\frac{1}{\sqrt{\kappa}} A_p[m,Q]-\sqrt{\kappa}B_p[\bar m,\bar Q],
\]
with $A_p$ and $B_p$ given in~(\ref{eq:Ap},\ref{eq:Bp}).
Assuming the small perturbations
\begin{align*}
m_\alpha^a = m_\alpha +\epsilon_\alpha^a\qquad\qquad
\bar m_\alpha^a = \bar m_\alpha +\bar \epsilon_\alpha^a\\[0.2cm]
Q_{ab} = q +\eta_{ab}\qquad\qquad\bar Q_{ab} = \bar q + \bar\eta_{ab},
\end{align*}
around the saddle point $(m_\alpha,\bar m_\alpha,q,\bar q)$, the perturbed free energy reads
\begin{align*}
\Delta f &= \sum_{a,\alpha}w_\alpha\bar\epsilon_\alpha^a\epsilon_\alpha^a+\frac{\sigma^2}{2}\sum_{a\ne b}\bar\eta_{ab}\eta_{ab}
+\sum_{a,b,\alpha,\beta}\bigl[\bigl(\delta_{ab}\bar A_{\alpha\beta}+\bar\delta_{ab}\bar B_{\alpha\beta}\bigr)\epsilon_\alpha^a\epsilon_\beta^b
+CT\bigr]\\[0.2cm]
&+\sum_{a\ne b,c,\alpha}\bigl[\bigl((\delta_{ab}+\delta_{ac})\bar C_{\alpha}+(1-\delta_{ac}-\delta_{bc})\bar D_{\alpha}\bigr)\epsilon_\alpha^c\eta_{ab}
+CT\bigr]\\[0.2cm]
&+\sum_{a\ne b,c\ne d}\bigl[\bigl(\delta_{(ab)(cd)}\bar E_0+\ind{a\in(cd)\oplus b\in(cd)}\bar E_1+\ind{(ab)\cap(cd)=\emptyset}\bar E_2\bigr)\eta_{ab}\eta_{cd}
+CT\bigr],
\end{align*}
where $CT$ means ``conjugate term'' in the sense $\epsilon \leftrightarrow \bar\epsilon$, $A_{\alpha\beta} \leftrightarrow \bar A_{\alpha\beta}$\ldots,
where $\bar\delta_{ab} \egaldef 1-\delta_{ab}$ and the operators are given by
\begin{align*}
A_{\alpha\beta} &\egaldef (\delta_{\alpha\beta}-m_\alpha m_\beta)w_\alpha w_\beta\qquad\qquad
B_{\alpha\beta} \egaldef \Bigl(\EE_{x,v}\bigl(v^\alpha v^\beta\tanh^2(\bar h(x,v))\bigr)-m_\alpha m_\beta\Bigr)w_\alpha w_\beta\\[0.2cm]
C_\alpha &\egaldef \frac{\kappa^{1/4}\sigma^2}{2}m_\alpha(1-q)w_\alpha\qquad\qquad
D_\alpha \egaldef \frac{\kappa^{1/4}\sigma^2}{2} \Bigl(\EE_{x,v}\bigl(v^\alpha \tanh^3(\bar h(x,v))\bigr)-m_\alpha q\Bigr)w_\alpha\\[0.2cm]
E_0 &\egaldef \frac{\sqrt{\kappa}\sigma^4}{4}(1-q^2)\qquad
E_1 \egaldef \frac{\sqrt{\kappa}\sigma^4}{4}q(1-q)\qquad
E_2 \egaldef \frac{\sqrt{\kappa}\sigma^4}{4}\Bigl(\EE_{x,v}\bigl(\tanh^4(\bar h(x,v))\bigr)-q^2\Bigr)
\end{align*}
with
\[
h(x,u) \egaldef \kappa^{1/4}\bigl(\sqrt{q}\sigma x + \sum_\alpha(m_\alpha w_\alpha - \eta_\alpha)u^\alpha\bigr),
\]
Conjugate quantities are obtained by replacing $m_\alpha$ by $\bar m_\alpha$, $q$ by $\bar q$, $u^\alpha$ by $v^\alpha$, $\eta_\alpha$ by $\theta_\alpha$ and $\kappa$ by $1/\kappa$.
As for the SK model, the $2Kp\times 2Kp$ Hessian thereby defined can be diagonalized with the help of three similar
sets of eigenmodes corresponding to different permutation symmetries in replica space.

The first set corresponds to $2K+2$ replica symmetric modes defined by $\eta_\alpha^a = \eta_\alpha$ and $\eta_{ab} = \eta$
solving the linear system
\begin{align*}
&\bigl(\frac{w_\alpha}{2}-\lambda\bigr)\bar\epsilon_\alpha-\frac{1}{2}\bar A_{\alpha\alpha}\epsilon_\alpha+\sum_\beta\bigl(\bar A_{\alpha\beta}
+(p-1)\bar B_{\alpha\beta}\bigr)\epsilon_\beta+\bigl((p-1)\bar C_\alpha+\frac{(p-1)(p-2)}{2}\bar D_\alpha\bigr)\eta = 0\\[0.2cm]
&\bigl(\frac{w_\alpha}{2}-\lambda\bigr)\epsilon_\alpha-\frac{1}{2}A_{\alpha\alpha}\bar\epsilon_\alpha+\sum_\beta\bigl(A_{\alpha\beta}
+(p-1)B_{\alpha\beta}\bigr)\bar\epsilon_\beta+\bigl((p-1)C_\alpha+\frac{(p-1)(p-2)}{2}D_\alpha\bigr)\bar\eta = 0\\[0.2cm]
&\bigl(\frac{\sigma^2}{2}-\lambda\bigr)\bar\eta+\sum_\alpha\bigl(\bar C_\alpha+\frac{p-2}{2}\bar D_\alpha\bigr)\epsilon_\alpha
+2\bigl(\bar E_0+2(p-2)\bar E_1+\frac{(p-2)(p-3)}{2}\bar E_2\bigr)\eta = 0\\[0.2cm]
&\bigl(\frac{\sigma^2}{2}-\lambda\bigr)\eta+\sum_\alpha\bigl(C_\alpha+\frac{p-2}{2}D_\alpha\bigr)\bar \epsilon_\alpha
+2\bigl(E_0+2(p-2)E_1+\frac{(p-2)(p-3)}{2}E_2\bigr)\bar \eta = 0
\end{align*}
with eigenvalue $\lambda$ solving a polynomial equation of degree $2K+2$ corresponding to a vanishing determinant in the above system.

The second set corresponds to a broken replica symmetry where one replica $a_0$ is different from the others
\[
(\epsilon_\alpha^a,\bar\epsilon_\alpha^a) =
\begin{cases}
(\epsilon_\alpha,\bar\epsilon_\alpha)\qquad\text{for}\ a\ne a_0\\[0.2cm]
(1-p)(\epsilon_\alpha,\bar\epsilon_\alpha)\qquad\text{for}\  a=a_0
\end{cases}
\qquad
(\eta_{ab},\bar\eta_{ab}) =
\begin{cases}
(\eta,\bar \eta)\qquad\text{for}\  a,b\ne a_0\\[0.2cm]
(1-\frac{p}{2})(\eta,\bar\eta)\qquad\text{for}\  a=a_0\ or\ b=a_0
\end{cases}
\]
This set has dimension $(2K+2)(p-1)$. Its parameterization is obtained by imposing orthogonality with the previous one.
The corresponding system reads
\begin{align*}
&\bigl(\frac{w_\alpha}{2}-\lambda\bigr)\bar\epsilon_\alpha-\frac{1}{2}\bar A_{\alpha\alpha}\epsilon_\alpha
+\sum_\beta(\bar A_{\alpha\beta}-\bar B_{\alpha\beta})\epsilon_\beta
+\frac{p-2}{2}\bigl(\bar C_\alpha-\bar D_\alpha\bigr)\eta = 0\\[0.2cm]
&\bigl(\frac{w_\alpha}{2}-\lambda\bigr)\epsilon_\alpha-\frac{1}{2}A_{\alpha\alpha}\bar\epsilon_\alpha+\sum_\beta(A_{\alpha\beta}-B_{\alpha\beta})\bar\epsilon_\beta
+\frac{p-2}{2}\bigl(C_\alpha-D_\alpha\bigr)\bar\eta = 0\\[0.2cm]
&\bigl(\frac{\sigma^2}{2}-\lambda\bigr)\bar\eta+\sum_\alpha(\bar C_\alpha-\bar D_\alpha)\epsilon_\alpha
+2\bigl(\bar E_0+(p-4)\bar E_1-(p-3)\bar E_2\bigr)\eta = 0\\[0.2cm]
&\bigl(\frac{\sigma^2}{2}-\lambda\bigr)\eta+\sum_\alpha(C_\alpha-D_\alpha)\bar\epsilon_\alpha
+2\bigl(E_0+(p-4)E_1-(p-3)E_2\bigr)\bar\eta = 0
\end{align*}

Finally the eigenmodes of the Hessian are made complete by considering a broken symmetry where two replicas $a_0$ and $a_1$ are different
from the others, with the following parameterization dictated again by orthogonality constraints with the previous sets:
\[
(\epsilon_\alpha^a,\bar\epsilon_\alpha^a) = 0,
\qquad
(\eta_{ab},\bar\eta_{ab}) =
\begin{cases}
(\eta,\bar \eta)\qquad\text{for}\ a,b\ne a_0\\[0.2cm]
\frac{3-p}{2}(\eta,\bar\eta)\qquad\text{for}\  a\in{a_0,a_1}\ or\ b\in{a_0,a_1}\\[0.2cm]
\frac{(p-2)(p-3)}{2}(\eta,\bar\eta)\qquad\text{for}\  (a,b)=(a_0,a_1).
\end{cases}
\]
The dimension of this set is now $p(p-3)$, and it represents eigenvectors iff the following system of equations is satisfied
\begin{align*}
&\bigl(\frac{\sigma^2}{2}-\lambda\bigr)\bar\eta +2(\bar E_0-2\bar E_1+\bar E_2)\eta = 0\\[0.2cm]
&\bigl(\frac{\sigma^2}{2}-\lambda\bigr)\eta +2(E_0-2E_1+E_2)\bar\eta = 0
\end{align*}
The corresponding eigenvalues read
\[
\lambda = \frac{\sigma^2}{2}\pm 2\sqrt{(\bar E_0-2\bar E_1+\bar E_2)(E_0-2E_1+E_2)},
\]
with degeneracy $p(p-3)/2$.
Finally the RS stability condition reads
\[
\frac{1}{\sigma^2} > \sqrt{\EE_{x,u}\Bigl(\sech^4\bigl(h(x,u)\bigr)\Bigr)\EE_{x,v}\Bigl(\sech^4\bigl(\bar h(x,v)\bigr)\Bigr)},
\]
which reduces to the same form of the AT line for the SK model when $\kappa=1$, except for the $u$ and $v$ averages that are specific to our model. As seen in Figure~\ref{fig:phasediag} the influence of $\kappa$ is very limited.

\section{Synthetic dataset}\label{app:data}
The multimodal distribution modeling the N-dimensional synthetic data is

\begin{equation}
P(s) = \sum_{c=1}^C p_c\prod_{i=1}^N \frac{e^{h_i^c s_i}}{2\cosh(h_i^c)},
\end{equation}
where \(C\) is the number of clusters, \(p_c\) is a weight and \(\bm{h}^c\) is a hidden field for cluster \(c\). The values for \(p_c\) are taken at random and normalized, while to compute \(h_i^c\) we take into account the magnetizations \(m_i^c = \tanh (h_i^c)\). Expanding over the spectral modes, we can set an effective dimension \(d\) by constraining the sum to the range \(\alpha = 1, \dots , d \)

\begin{equation}
m_i^c = \sum_{\alpha = 1}^d m_{\alpha}^c u_i^\alpha
\end{equation}
Clusters' magnetizations \(m_{\alpha}^c\) are drawn at random between \([-1, 1]\) and normalized with the factor

\begin{equation}
Z = \sqrt{\frac{\sum_{\alpha} m_{\alpha}^2}{d \cdot r}}, \quad r = \tanh (\eta)
\end{equation}
where \(r\) is introduced to decrease the clusters' polarizations (in  our simulations, we used \(\eta = 0.3\)). The spectral basis \( u_i^\alpha \) is obtained by drawing at random \(d\) N-dimensional vectors and applying the Gram-Schmidt process (which can be safely employed as N is supposedly big and thus the initial vectors are nearly orthogonal). The hidden fields are then obtained from the magnetizations

\begin{equation}
h_i^c = \tanh^{-1}(m_i^c)
\end{equation}
and the samples are generated by choosing a cluster according to \(p_c\) and setting the visible variables to \( \pm 1\) according to

\begin{equation}
p(s_i = 1) = \frac{1}{1 + e^{-2 h_i^c}}
\end{equation}

\bibliographystyle{unsrt}
\bibliography{rbm}

\end{document}

%% file: rbm.pdf_t
\begin{picture}(0,0)%
\includegraphics{rbm.pdf}%
\end{picture}%
\setlength{\unitlength}{4144sp}%
\begingroup\makeatletter\ifx\SetFigFont\undefined%
\gdef\SetFigFont#1#2#3#4#5{%
  \reset@font\fontsize{#1}{#2pt}%
  \fontfamily{#3}\fontseries{#4}\fontshape{#5}%
  \selectfont}%
\fi\endgroup%
\begin{picture}(9072,3249)(1246,-3046)
\put(1261,-2896){\makebox(0,0)[lb]{\smash{{\SetFigFont{29}{34.8}{\rmdefault}{\mddefault}{\updefault}$s_1$}}}}
\put(8011,-2896){\makebox(0,0)[lb]{\smash{{\SetFigFont{29}{34.8}{\rmdefault}{\mddefault}{\updefault}$s_{\textsc{n}_v}$}}}}
\put(2386,-196){\makebox(0,0)[lb]{\smash{{\SetFigFont{29}{34.8}{\rmdefault}{\mddefault}{\updefault}$\sigma_1$}}}}
\put(6436,-241){\makebox(0,0)[lb]{\smash{{\SetFigFont{29}{34.8}{\rmdefault}{\mddefault}{\updefault}$\sigma_{\textsc{n}_h}$}}}}
\put(4366,-2896){\makebox(0,0)[lb]{\smash{{\SetFigFont{29}{34.8}{\rmdefault}{\mddefault}{\updefault}$s_i$}}}}
\put(4591,-241){\makebox(0,0)[lb]{\smash{{\SetFigFont{29}{34.8}{\rmdefault}{\mddefault}{\updefault}$\sigma_j$}}}}
\put(8596,-466){\makebox(0,0)[lb]{\smash{{\SetFigFont{20}{24.0}{\rmdefault}{\mddefault}{\updefault}Hidden layer}}}}
\put(8596,-2491){\makebox(0,0)[lb]{\smash{{\SetFigFont{20}{24.0}{\rmdefault}{\mddefault}{\updefault}Visible layer}}}}
\put(4546,-1501){\makebox(0,0)[lb]{\smash{{\SetFigFont{50}{60.0}{\rmdefault}{\mddefault}{\updefault}$W_{ij}$}}}}
\end{picture}%